\documentclass{siamart0516}

\usepackage{lipsum}
\usepackage{amsfonts}
\usepackage{graphicx}
\usepackage{epstopdf}
\usepackage{algorithmic}
\ifpdf
  \DeclareGraphicsExtensions{.eps,.pdf,.png,.jpg}
\else
  \DeclareGraphicsExtensions{.eps}
\fi

\usepackage{amsmath,amssymb,amsfonts,latexsym,epsfig,graphicx}
\usepackage{dsfont}
\usepackage{hyperref}

\newcommand{\TheTitle}{The Complexity of General-Valued CSPs} 
\newcommand{\TheAuthors}{V. Kolmogorov, A. Krokhin, and M. Rol\'inek}

\headers{\TheTitle}{\TheAuthors}

\title{{\TheTitle}\thanks{Accepted to SIAM Journal on Computing (SICOMP). An extended abstract of this work (without proofs) has appeared in FOCS 2015.
\funding{The first and third authors are supported by the European Research Council under the European Unions Seventh Framework Programme (FP7/2007-2013)/ERC grant agreement no 616160.}}}

\author{
  Vladimir Kolmogorov\thanks{IST Austria (\email{vnk@ist.ac.at}, \email{michal.rolinek@ist.ac.at}).}
  \and
  Andrei Krokhin\thanks{Durham University, UK (\email{andrei.krokhin@durham.ac.uk}).}
  \and
  Michal Rol\'inek\footnotemark[2]
}

\usepackage{amsopn}

\DeclareMathOperator{\supp}{supp}
\DeclareMathOperator{\dom}{dom}
\DeclareMathOperator{\Opt}{Opt}
\DeclareMathOperator{\BLP}{BLP}

\newcommand{\q}{\ensuremath{\mathbb{Q}}}

\newcommand{\Pol}[1]{\ensuremath{\operatorname{Pol}(#1)}}

\newcommand{\fPol}[1]{\ensuremath{\operatorname{fPol}(#1)}}
\newcommand{\fPolplus}[1]{\ensuremath{\operatorname{fPol}^+(#1)}}
\DeclareMathOperator{\pol}{Pol}

\DeclareMathOperator{\CostF}{{\cal F}}

\newcommand{\CSP}[1]{\ensuremath{\operatorname{CSP}(#1)}}

\newcommand{\vcsp}[1]{\ensuremath{\operatorname{VCSP}(#1)}}
\newcommand{\VCSP}[1]{\ensuremath{\operatorname{VCSP}(#1)}}

\newcommand{\Feas}[1]{\ensuremath{\operatorname{Feas}(#1)}}

\long\def\ignore#1{}
\def\myps[#1]#2{\includegraphics[#1]{#2}}

\def\br(#1,#2){{\langle #1,#2 \rangle}}
\def\setZ[#1,#2]{{[ #1 .. #2 ]}}

\def\zd{,\ldots,}

\newcommand{\Qc}{\mbox{$\overline{\mathbb Q}$}}

\def\closure#1{{\langle#1\rangle}}

\def\q={\quad=\quad}
\def\qq={\qquad=\qquad}

\def\calC{{\cal C}}

\def\calI{{\cal I}}

\def\calV{{\cal V}}
\def\calO{{\cal O}}

\def\psfile[#1]#2{}
\def\psfilehere[#1]#2{}

\def\assign(#1,#2){\langle#1,#2\rangle}
\def\edge(#1,#2){(#1,#2)}

\def\slack(#1){\texttt{slack}({#1})}
\def\barslack(#1){\overline{\texttt{slack}}({#1})}

\def\unitvec(#1){{{\bf u}_{#1}}}

\newtheorem{example}[theorem]{Example}
\newtheorem{conjecture}[theorem]{Conjecture}

\ifpdf
\hypersetup{
  pdftitle={\TheTitle},
  pdfauthor={\TheAuthors}
}
\fi

\begin{document}

\maketitle

% REQUIRED
\begin{abstract}
An instance of the Valued Constraint Satisfaction Problem (VCSP) is given by a finite
set of variables, a finite domain of labels, and a sum of functions,
each function depending on a subset of the variables. Each function can take finite values specifying costs
of assignments of labels to its variables or
the infinite value, which indicates an infeasible assignment. The goal is to find
an assignment of labels to the variables that minimizes the sum.

We study, assuming that P $\ne$ NP, how the complexity of this
very general problem depends on the set of functions allowed in the instances, the so-called constraint language.
The case when all allowed functions take values in $\{0,\infty\}$ corresponds to ordinary CSPs, where one deals only with the feasibility issue
and there is no optimization. This case is the subject of the Algebraic CSP Dichotomy Conjecture predicting for which constraint languages
CSPs are tractable (i.e. solvable in polynomial time) and for which NP-hard. The case when all allowed functions take only finite values
corresponds to finite-valued CSP, where the feasibility aspect is trivial and
one deals only with the optimization  issue.
The complexity of finite-valued CSPs was fully classified by Thapper and \v{Z}ivn\'y.

An algebraic necessary condition for tractability of a general-valued CSP with a fixed constraint language was recently given by Kozik and Ochremiak. As our main result, we prove that if a constraint language satisfies this algebraic necessary condition, and
the feasibility CSP (i.e. the problem of deciding whether a given instance has a feasible solution) corresponding to the VCSP with this language is tractable,
then the VCSP is tractable.
The algorithm is a simple combination of the assumed algorithm for the feasibility CSP and the standard LP relaxation.
As a corollary, we obtain that a dichotomy for ordinary CSPs would imply a dichotomy for general-valued CSPs.
\end{abstract}

% REQUIRED

% REQUIRED

\section{Introduction}\label{sec:intro}

Computational problems from many different areas involve finding an assignment
of labels to a set of variables, where that assignment must satisfy some specified
feasibility conditions and/or optimize some specified objective function.
In many such problems, the feasibility conditions are local and also the objective function can be represented as a sum of functions,
each of which depends on some subset of the variables. Examples include:
% Linear and Integer Programming,
Gibbs energy minimization, Markov Random Fields (MRF),
Conditional Random Fields (CRF), Min-Sum Problems,
Minimum Cost Homomorphism,
Constraint Optimization Problems (COP) and Valued Constraint Satisfaction Problems
(VCSP)~\cite{Boykov98:markov,Crama11:book,Lauritzen96,Rossi06:handbook,Wainwright08:graphical},.

The constraint satisfaction problem  provides a common framework for many
theoretical and practical problems in computer science~\cite{Creignou01:book,Rossi06:handbook}. An
instance of the \emph{constraint satisfaction problem} (CSP) consists of a
collection of variables that must be assigned labels from a given domain subject
to specified constraints~\cite{Montanari74:constraints}. The CSP is equivalent
to the problem of evaluating conjunctive queries on
databases~\cite{Kolaitis00:conjunctive}, and to the homomorphism problem for relational
structures~\cite{Feder98:monotone}. The CSP deals only with the feasibility
issue: can all constraints be satisfied simultaneously?

There are several natural optimization versions of the CSP:
{\sc Max CSP} (or {\sc Min CSP}) where the goal is to find the assignment maximizing the number of satisfied constraints (or minimizing the number
of unsatisfied constraints)~\cite{Cohen05:supermodular,Creignou01:book,Jonsson06:maxcsp3,Jonsson11:mincsp4},
problems like {\sc Max-Ones} and {\sc Min-Hom} where the constraints must be satisfied and some additional function
of the assignment is to be optimized~\cite{Creignou01:book,Jonsson08:maxsol,Takhanov10:dichotomy},
and, the most general version, {\em valued CSP} or VCSP (also known as soft CSP),  where each combination of values for variables in a constraint has a cost
and the goal is to minimize the aggregate cost~\cite{Cohen13:algebraic,cohen06:complexitysoft,kolmogorov15:power,tz13:stoc}.
Thus, an instance of the VCSP amounts to minimizing a sum of functions, each depending on a subset of variables.
By using infinite costs to indicate infeasible combinations,  VCSP can model both feasibility and optimization aspects and so considerably generalises all the problems mentioned above~\cite{Cohen13:algebraic,cohen06:complexitysoft,Krokhin17:new_survey}.
There is much activity and very strong results concerning various aspects of approximability of (V)CSPs (see e.g.~\cite{Barto16SIAM:robust,Brown-Cohen:ICALP16,Chan16:approximate,Creignou01:book,Deineko08:constants,Ene13:local,Hastad08:2csp,Raghavendra08:optimal} for a small sample), but in this paper we focus on solving VCSPs to optimality.

We assume throughout the paper that P $\ne$ NP.
Since all the above problems are NP-hard in full generality, a major line of research in CSP tries to identify
the tractable cases of such problems (see books/surveys~\cite{Cohen06:handbook,Creignou01:book,Creignou08:complexity,Krokhin17:new_survey}), the primary motivation being the general picture
rather than specific applications.
The two main ingredients of a constraint are (a) variables to which it is applied and (b) relations/functions specifying
the allowed combinations of values or the costs for all combinations.
Therefore, the main types of restrictions on CSP are (a) {\em structural} where
the hypergraph formed by sets of variables appearing in individual constraints is restricted~\cite{Gottlob09:tractable,Marx13:tractable}, and (b) {\em language-based} where the constraint language, i.e. the set of relations/functions that can appear in constraints, is fixed (see, e.g.~\cite{bulatov05:classifying,Cohen06:handbook,Creignou01:book,Feder98:monotone,tz13:stoc}). The ultimate sort of results in these directions are {\em dichotomy} results, pioneered by~\cite{Schaefer78:complexity},
which characterise the tractable restrictions and show that the rest are as hard as the corresponding general problem
(which cannot generally be taken for granted).
The language-based direction is considerably more active than the structural one, there are many partial language-based dichotomy results, e.g.~\cite{bulatov06:3-elementjacm,Bulatov11:conservative,cohen06:complexitysoft,Creignou01:book,Jonsson06:maxcsp3,Jonsson11:mincsp4,Kolmogorov13:conservative,Takhanov10:dichotomy}, but many central questions are still open.
In this paper, we study VCSPs with a fixed constraint language on a finite domain, and all further discussion concerns only such CSPs and VCSPs.

{\bf Related Work.} The CSP Dichotomy Conjecture, stating that each CSP is either tractable or NP-hard, was first formulated by Feder and Vardi~\cite{Feder98:monotone}. The universal-algebraic approach to this problem was discovered in~\cite{bulatov05:classifying,Jeavons98:algebraic,Jeavons97:closure}, and the precise boundary between the tractable cases and NP-hard cases was conjectured in algebraic terms in~\cite{bulatov05:classifying}, in what is now known as the Algebraic CSP Dichotomy Conjecture (see Conjecture~\ref{con:AlgCSPDich}). The hardness part was proved in~\cite{bulatov05:classifying}, and it is the tractability part that is the essence of the conjecture. This conjecture is still open in full generality and is the object of much investigation, e.g.~\cite{Barto16:collapse,Barto12:absorbing,Barto14:jacm,barto11:lics,barto09:siam,bulatov05:classifying,Bulatov11:conservative,Cohen06:handbook,Idziak10:tractability}. It is known to hold for domains with at most 3 elements~\cite{bulatov06:3-elementjacm,Schaefer78:complexity}, for smooth digraphs~\cite{barto09:siam}, and for the case when all unary relations are available~\cite{barto11:lics,Bulatov11:conservative}.  The main two polynomial-time algorithms used for CSPs are based one on local consistency (``bounded width'') and the other on compact representation of solution sets (``few subpowers''), and their applicability (in pure form) is fully characterized in~\cite{Barto16:collapse,Barto14:jacm} and~\cite{Idziak10:tractability}, respectively.

At the opposite (to CSP) end of the VCSP spectrum are the finite-valued CSPs, in which functions do not take infinite values. In such VCSPs, the feasibility aspect is trivial, and one has to deal only with the optimization issue.
One polynomial-time algorithm that solves tractable finite-valued CSPs is based on the so-called basic linear programming (BLP) relaxation,
and its applicability (also for the general-valued case) was fully characterized in~\cite{kolmogorov15:power} (see Theorem~\ref{thm:BLP}).  The complexity of finite-valued CSPs was completely classified in~\cite{tz13:stoc}, where it is shown that all finite-valued CSPs not solvable by BLP are NP-hard.

For general-valued CSPs, full classifications are known for the Boolean case (i.e., when the domain is two-element)~\cite{cohen06:complexitysoft}
and also for the case when all 0-1-valued unary cost functions are available~\cite{Kolmogorov13:conservative}.
The algebraic approach to the CSP was extended to VCSPs in~\cite{Cohen13:algebraic,Cohen06:algebraic,cohen06:complexitysoft,Kozik15:ICALP}, and was also key to
much progress.
An algebraic necessary condition for a VCSP to be tractable was recently proved by Kozik and Ochremiak in~\cite{Kozik15:ICALP}, where this condition was also conjectured to be sufficient (see Theorem~\ref{thm:VCSPhard} and Conjecture~\ref{con:VCSPtractable} below).
This conjecture can  be called the Algebraic VCSP Dichotomy Conjecture, and it is a generalization of the corresponding conjecture for CSP.
A large family of VCSPs satisfying the necessary condition from~\cite{Kozik15:ICALP} has recently been shown tractable via a low-level Sherali-Adams hierarchy relaxation~\cite{Thapper15:Sherali}.

Our proof uses the technique of ``lifting a language'' introduced in~\cite{Kolmogorov2015:hybrid}.

{\bf Our Contribution.} We completely classify the complexity of VCSPs with a fixed constraint language modulo the complexity of CSPs
(see Theorem~\ref{thm:main}). Clearly, for a VCSP to be tractable, it is necessary that the corresponding feasibility CSP is tractable. We prove that any VCSP satisfying this necessary condition and the necessary condition of Kozik and Ochremiak is tractable. The polynomial-time algorithm that solves such VCSP is a simple combination of the (assumed) polynomial-time algorithm for the feasibility CSP and BLP  (see Theorem~\ref{thm:main-alg}). Thus, our dichotomy theorem generalizes the dichotomy for finite-valued CSPs from~\cite{tz13:stoc}, and, with the help of the CSP tractability result from~\cite{Barto14:jacm}, it also implies the tractability of VCSPs
shown tractable in~\cite{Thapper15:Sherali,Thapper16:powerSA}.

Our classification result has the following several unexpected features. One is that the algorithm that solves all tractable VCSPs uses 
feasibility checking only as a black-box. The other is that the algorithm is simply feasibility preprocessing followed by BLP - this was unexpected, for example, because higher levels of the Sherali-Adams hierarchy were used in~\cite{Thapper15:Sherali} to prove tractability of a wide class
of VCSPs. Finally, the proof of our result avoids structural universal algebra present in most CSP classifications and in~\cite{Kozik15:ICALP, Kozik15:algebraic}. 

Our result says that any dichotomy for CSP (not necessarily the one predicted by the Algebraic CSP Dichotomy Conjecture)
will imply a dichotomy for VCSP.
However, if the Algebraic CSP Dichotomy Conjecture holds then the necessary algebraic condition of Kozik and Ochremiak
guarantees tractability of the feasibility CSP (see~\cite{Kozik15:ICALP}), implying that this algebraic condition alone is necessary and sufficient for tractability of a VCSP, and also that all the intractable VCSPs are NP-hard.
In particular, the Algebraic CSP Dichotomy Conjecture implies the Algebraic VCSP Dichotomy Conjecture.

On the technical level, some of our proofs (e.g.\ those in Section~\ref{sec:thsym}) use techniques established in~\cite{kolmogorov15:power,tz13:stoc}, while others (e.g.\ all of Section~\ref{sec:spec}) introduce new technical ideas.

Our result is the culmination of research into complexity classification of language-based VCSPs in the sense that its scope cannot be widened, the yet unclassified part of the VCSP landscape is the (non-valued) CSP. One could, of course, extend the classification framework by looking at
other forms of algorithmic tractability, say, approximation algorithms or fixed-parameter tractability, and such extensions will have many open questions. It is also interesting to obtain tighter and more explicit characterisations for important special cases of VCSP (as done in~\cite{Thapper16:powerSA}, for example), by deriving them from our main result or otherwise.

%%%%%%%%%%%%%%%%%%%%%%%%%%%%%%%%%%%%%%%%%%%%%%%%%%%%%%%%%%%%%%%%%%%%%%%%%%%%%%%%%%%%%%%%%%%%%%%%%%%%%%%%%%%%%%
%%%%%%%%%%%%%%%%%%%%%%%%%%%%%%%%%%%%%%%%%%%%%%%%%%%%%%%%%%%%%%%%%%%%%%%%%%%%%%%%%%%%%%%%%%%%%%%%%%%%%%%%%%%%%%
%%%%%%%%%%%%%%%%%%%%%%%%%%%%%%%%%%%%%%%%%%%%%%%%%%%%%%%%%%%%%%%%%%%%%%%%%%%%%%%%%%%%%%%%%%%%%%%%%%%%%%%%%%%%%%
%%%%%%%%%%%%%%%%%%%%%%%%%%%%%%%%%%%%%%%%%%%%%%%%%%%%%%%%%%%%%%%%%%%%%%%%%%%%%%%%%%%%%%%%%%%%%%%%%%%%%%%%%%%%%%
%%%%%%%%%%%%%%%%%%%%%%%%%%%%%%%%%%%%%%%%%%%%%%%%%%%%%%%%%%%%%%%%%%%%%%%%%%%%%%%%%%%%%%%%%%%%%%%%%%%%%%%%%%%%%%
%%%%%%%%%%%%%%%%%%%%%%%%%%%%%%%%%%%%%%%%%%%%%%%%%%%%%%%%%%%%%%%%%%%%%%%%%%%%%%%%%%%%%%%%%%%%%%%%%%%%%%%%%%%%%%
%%%%%%%%%%%%%%%%%%%%%%%%%%%%%%%%%%%%%%%%%%%%%%%%%%%%%%%%%%%%%%%%%%%%%%%%%%%%%%%%%%%%%%%%%%%%%%%%%%%%%%%%%%%%%%
%%%%%%%%%%%%%%%%%%%%%%%%%%%%%%%%%%%%%%%%%%%%%%%%%%%%%%%%%%%%%%%%%%%%%%%%%%%%%%%%%%%%%%%%%%%%%%%%%%%%%%%%%%%%%%
%%%%%%%%%%%%%%%%%%%%%%%%%%%%%%%%%%%%%%%%%%%%%%%%%%%%%%%%%%%%%%%%%%%%%%%%%%%%%%%%%%%%%%%%%%%%%%%%%%%%%%%%%%%%%%

%%%%%%%%%%%%%%%%%%%%%%%%%%%%%%%%%%%%%%%%%%%%%%%%%%%%%%%%%%%%%%%%%%%%%%%%%%%

\section{Preliminaries}

\subsection{Valued Constraint Satisfaction Problems}\label{sec:VCSPdef}

Throughout the paper, let $D$ be a fixed finite set and let $\Qc=\mathbb{Q}\cup\{\infty\}$ denote the set
of rational numbers with (positive) infinity.

\begin{definition}
We denote the set of all functions $f:D^n \rightarrow \Qc$ by $\CostF_D^{(n)}$ and let $\CostF_D=\bigcup_{n\ge 1}{\CostF_D^{(n)}}$.
We will often call the functions in $\CostF_D$ {\em cost functions} over $D$.
%{\em feasibility relation} defined as $\mathrm{Feas}(\phi)=\{\vec{x}\in D^m\mid \phi(\vec{x}) \mbox{ is finite}\}$.
For every cost function $f\in\CostF_D^{(n)}$, let $\dom f=\{x\mid f(x)<\infty\}$. Note that $\dom f$ can be considered both as an $n$-ary relation
and as a $n$-ary function such that $\dom f(x)=0$ if and only if $f(x)$ is finite.
\end{definition}

We will call the
set $D$ the \emph{domain},
elements of $D$ \emph{labels} (for variables),
and say that the cost functions in $\CostF_D$ take \emph{values}.
Note that in some papers on VCSP, e.g.~\cite{Cohen13:algebraic,Thapper15:Sherali}, cost functions are called weighted relations.

\begin{definition}
An instance of the {\em valued constraint satisfaction problem} (VCSP)
is a function from $D^V$ to $\Qc$ given by
\begin{equation}\label{eq:VCSPinst}
f_{\calI}(x)=\sum_{t\in T}{f_t(x_{v(t,1)},\ldots,x_{v(t,n_t)})},
\end{equation}
where $V$ is a finite set of variables, $T$ is a finite set of constraints, each constraint is specified
by a cost function $f_t$ of arity $n_t$ and indices $v(t,k)$, $k=1,\ldots,n_t$.
The goal is to find an {\em assignment} (or {\em labeling}) $x\in D^V$ that minimizes $f_{\calI}$.
The value of an optimal assignment is denoted by $\Opt(\calI)$.
\end{definition}

\begin{definition}
Any set $\Gamma\subseteq \CostF_D$ is called
a \emph{valued constraint language} over $D$, or simply a \emph{language}.
We will denote by \vcsp{\Gamma} the class of all VCSP instances in which the
constraint functions $f_t$ are all contained in $\Gamma$.
Instances of \vcsp{\Gamma} will sometimes be called just \emph{$\Gamma$-instances}.
%We assume that each function $f_t$ is given by the full table of values.
\end{definition}

This framework subsumes many other frameworks studied earlier and captures many specific well-known problems,
including {\sc $k$-Sat}, {\sc Graph $k$-Colouring}, {\sc Max Cut}, {\sc Min Vertex Cover} and others (see~\cite{Krokhin17:new_survey}).
Note that if every function in $\Gamma$ takes values in $\{0,\infty\}$ (such functions are often called {\em crisp}) then $\VCSP\Gamma$ is a pure feasibility problem, commonly known as $\CSP\Gamma$.

The main goal of our line of research is to classify the complexity of problems $\VCSP\Gamma$. Problems $\CSP\Gamma$ and $\VCSP\Gamma$ are
called tractable if, for each finite $\Gamma'\subseteq \Gamma$, $\VCSP{\Gamma'}$ is tractable. Also, $\VCSP\Gamma$ is called NP-hard if, for some finite $\Gamma'\subseteq \Gamma$, $\VCSP{\Gamma'}$ is NP-hard. One advantage of defining tractability in terms of finite subsets is that the
tractability of a valued constraint language is independent of whether the cost
functions are represented explicitly (say, via full tables of values, or via tables for
the finite-valued parts) or implicitly (via oracles). Following~\cite{bulatov05:classifying}, we say that $\VCSP\Gamma$ is {\em globally tractable}
there is a polynomial-time algorithm solving $\VCSP\Gamma$, assuming all functions in instances are given by full tables of values.
For CSPs, there is no example of $\CSP\Gamma$ that is tractable, but not globally tractable, and it is conjectured in~\cite{bulatov05:classifying}
that no such $\CSP\Gamma$ exists.

%We restrict ourselves only to finite languages because we rely on some results from~\cite{Kozik15:algebraic} that are at the moment known
%only for such languages. The tractability part of our dichotomy result, however, holds in the above sense for arbitrary languages.

\subsection{Polymorphisms, Expressibility, Cores}

Let $\calO_D^{(m)}$ denote the set of all operations $g:D^m\rightarrow D$ and let
 $\calO_D=\cup_{m\ge 1}{\calO_D^{(m)}}$. When $D$ is clear from the context, we will sometimes write simply $\calO^{(m)}$ and $\calO$.

Any language $\Gamma$ defined on $D$ can be associated
with a set of operations on $D$,
known as the polymorphisms of $\Gamma$, which allow one to combine (often in a useful way)
several feasible assignments into a new one.
\begin{definition}
\label{def:polymorphism}
An operation $g\in \calO_D^{(m)}$ is a \emph{polymorphism} of a cost function $f\in\CostF_D$ if,
for any $x^1,x^2,\ldots,x^m \in \dom f$,
we have that $g(x^1,x^2,\ldots,x^m)\in \dom f$ where $g$ is applied component-wise.

For any valued constraint language $\Gamma$ over a set $D$,
we denote by $\pol(\Gamma)$ the set of all operations on $D$ which are polymorphisms of every
$f \in \Gamma$.
%We denote by $\polk{\Gamma}$ the $k$-ary operations in $\pol(\Gamma)$.
\end{definition}
%
%\noindent
\begin{example} Let $f\in \CostF^{(n)}_{\{0,1\}}$ be such that $f(1,\ldots,1,0)=\infty$ and $f(a_1,\ldots,a_n)=0$ otherwise. It corresponds to the Horn clause $(x_1\vee \ldots \vee x_{n-1} \vee \overline{x_n})$. Then it is well known and easy to see that the binary operation $\min\in \calO_{\{0,1\}}$ is a polymorphism of $f$.
\end{example}

Clearly, if $g$ is a polymorphism of a cost function $f$, then $g$ is also a polymorphism of $\dom f$.
For $\{0,\infty\}$-valued functions, which naturally correspond to relations, the notion
of a polymorphism defined above coincides with the standard notion of a polymorphism for relations.
Note that the projections (aka dictators), i.e. operations of the form $e_n^i(x_1,\ldots,x_n)=x_i$,  are polymorphisms of all valued constraint languages.
Polymorphisms play the key role in the algebraic approach to the CSP, but, for VCSPs, more general constructs are necessary, which we now define.

\begin{definition}
An $m$-ary \emph{fractional operation} $\omega$ on $D$ is a probability distribution on $\calO_D^{(m)}$.
The support of $\omega$ is defined as $\supp(\omega)=\{g\in \calO_D^{(m)}\mid \omega(g)>0\}$.
\end{definition}

\begin{definition}\label{def:fpol}
A $m$-ary fractional operation $\omega$ on $D$ is said to be a \emph{fractional polymorphism} of a
cost function $f\in \CostF_D$ if, for any $x^1,x^2,\ldots,x^m \in \dom f$,
we have
\begin{equation}
\sum_{g\in\supp(\omega)}{\omega(g)f(g(x^1,\ldots,x^m))} \le \frac{1}{m}(f(x^1)+\ldots+f(x^m)).
\label{eq:wpol-dist}
\end{equation}

For a constraint language $\Gamma$, $\fPol\Gamma$ will denote the set of all fractional operations that are fractional polymorphisms
of each function in $\Gamma$. Also, let $\fPolplus\Gamma=\{g\in \calO_D\mid g\in\supp(\omega), \omega\in\fPol\Gamma\}$.
\end{definition}

The intuition behind the notion of fractional polymorphism is that it allows one to combine several feasible assignments
into new feasible assignments so that the expected value of a  new assignment (non-strictly) improves the average value of the original assignments.

\begin{example} % changed "If ..." to "Suppose ..." so that the text does not go over the right margin - vnk
Suppose that $\omega$ is a binary fractional operation on $D=\{0,1\}$ such that $\omega(\min)=\omega(\max)=1/2$. Then it is well-known and easy to check that the finite-valued functions with fractional polymorphism $\omega$ are the submodular functions.
Moreover, functions with this fractional polymorphism that are not necessarily finite-valued precisely correspond to submodular functions defined on a ring family.
\end{example}

More examples of fractional polymorphisms can be found in~\cite{Krokhin17:new_survey,kolmogorov15:power,tz13:stoc}.
%Clearly, we have $\fPolplus\Gamma\subseteq\Pol\Gamma$ for any $\Gamma$.

%The following notion is useful for the general algebraic theory of VCSPs~\cite{Cohen13:algebraic,Kozik15:algebraic}, but we will need it only to %connect some statements from~\cite{Kozik15:algebraic} with fractional polymorphisms.
%
%\begin{definition}\label{def:wpol}
%Let $C$ be a set of $m$-ary operations in $\Pol\Gamma$. A function $\varpi:C\rightarrow \mathbb{R}$ is called a {\em weighting} (of $C$) if
%$\sum_{g\in C}{\varpi(g)}=0$ and $\varpi(g)<0$ only if $g$ is a projection. A weighting $\varpi$ is called a {\em weighted polymorphism} of $\Gamma$
%if, for any $f\in \Gamma$ and any $x^1,\ldots,x^m\in \dom f$, it holds that
%\[
%\sum_{g\in C}{\varpi(g)\cdot f(g(x^1,\ldots,x^m))}\le 0.
%\]
%The support of $\varpi$ is $\supp(\varpi)=\{g\in C\mid \varpi(g)>0\}$.
%\end{definition}

We remark that, in some papers (e.g., in~\cite{Cohen13:algebraic}), fractional polymorphisms (and closely related
objects called weighted polymorphisms) are defined as rational-valued functions, which is sufficient for analysing the complexity
of VCSPs with finite constraint languages. However, real-valued fractional polymorphisms are necessary to analyse infinite constraint languages~\cite{Fulla:2016:Galois,Kozik15:algebraic,tz13:stoc}.

The key observation in the algebraic approach to (V)CSP is that
neither the complexity nor the algebraic properties of a language $\Gamma$ change when functions ``expressible'' from $\Gamma$ in a certain way are added to it.

\begin{definition}
For a constraint language $\Gamma$, let $\langle\Gamma\rangle$ denote the set of all functions $f(x_1\zd x_k)$ such that,
for some instance $\calI$ of $\VCSP\Gamma$ with objective function $f_{\calI}(x_1\zd x_k,x_{k+1}\zd x_n)$,
we have \[f(x_1\zd x_k)=\min_{x_{k+1}\zd x_n}{f_{\calI}(x_1\zd x_k,x_{k+1}\zd x_n)}.\]
We then say that $\Gamma$ {\em expresses} $f$, and call $\langle\Gamma\rangle$ the {\em expressive power} of $\Gamma$.
\end{definition}

\begin{lemma}[\cite{Cohen06:algebraic,cohen06:complexitysoft}]
\label{lem:exppower}
Let $f\in \langle\Gamma\rangle$. Then
\begin{enumerate}
\item if $\omega\in\fPol\Gamma$ then $\omega$ is a fractional polymorphism of $f$ and of $\dom f$;
\item $\VCSP\Gamma$ is tractable if and only if $\VCSP{\Gamma\cup\{f,\dom f\}}$ is tractable;
\item $\VCSP\Gamma$ is NP-hard if and only if $\VCSP{\Gamma\cup\{f,\dom f\}}$ is NP-hard.
\end{enumerate}
\end{lemma}

%\subsection{Cores and constants}

%A language $\Gamma$ is called a {\em core} if, for any unary $\omega\in\fPol\Gamma$, all operations in $\supp(\omega}$ are bijections.

The dichotomy problem for VCSPs can be reduced to a class of constraint languages called rigid cores, defined below.
Apart from reducing the cases that need to be considered, this reduction enabled the use of much more powerful results
from universal algebra than what can be done without this restriction (see, e.g.~\cite{Kozik15:algebraic}).

For a subset $D'\subseteq D$, let $u_{D'}$ be the function defined as follows: $u_{D'}(d)=0$ if $d\in D'$ and $u_{D'}(d)=\infty$ otherwise.
We write $u_d$ for $u_{\{d\}}$. Let $\calC_D=\{\{u_d\}\mid d\in D\}$.

\begin{lemma}[\cite{Kozik15:algebraic}]
For any valued constraint language $\Gamma'$ on a finite set $D'$, there is a subset $D\subseteq D'$ and a valued constraint language
$\Gamma$ on $D$ such that $\calC_D\subseteq \Gamma$ and
the problems $\VCSP{\Gamma'}$ and $\VCSP\Gamma$ are polynomial-time equivalent.
\end{lemma}

This language $\Gamma$ is called the \emph{rigid core} of $\Gamma'$, and it can be obtained from $\Gamma'$ as follows.
Let $g'$ be a unary operation on $D'$ with minimum $|g'(D')|$ among all unary operations $g'\in \fPolplus{\Gamma'}$. Then $D$ is set to be $g'(D')$ and $\Gamma$ is set to be $\{f|_D: f\in\Gamma'\}\cup\calC_D$.
Thus, the intuition behind moving to the rigid core is that (a) one removes labels from the domain that
can always be (uniformly) replaced in any solution to an instance without increasing its value, and (b) one allows constraints
of the form $u_d$ that can be used to fix labels for variables, leading to applicability of more powerful algebraic
results.

\subsection{Cyclic and symmetric operations}

Several types of operations play a special role in the algebraic approach to (V)CSP.

\begin{definition}
An operation $g\in \calO_D^{(m)}$, $m\ge 2$, is called
\begin{itemize}
\item \emph{idempotent} if $g(x,\ldots,x)=x$ for all $x\in D$;
\item \emph{Taylor} if, for each $1\le i\le m$, it satisfies an identity of the form $g(\vartriangle_1,\vartriangle_2,\ldots,\vartriangle_m)=g(\square_1,\square_2,\ldots,\square_m)$ where all $\vartriangle_j,\square_j$ are in $\{x,y\}$ and $\vartriangle_i\ne \square_i$.
\item \emph{cyclic} if $g(x_1,x_2,\ldots,x_m)=g(x_2,\ldots,x_m,x_1)$ for all $x_1,\ldots,x_m\in D$;
\item \emph{symmetric} if $g(x_1,x_2,\ldots,x_m)\!=\!g(x_{\pi(1)},x_{\pi(2)},\ldots,x_{\pi(m)})$ for all $x_1,\ldots,x_m\!\in\! D$, and any permutation $\pi$ on $[m]$.
\end{itemize}
A fractional operation $\omega$ is said to be idempotent/cyclic/symmetric if all operations in $\supp(\omega)$ have the corresponding property.
%Similarly, a weighting $\varpi$ with $\supp(\varpi)\neq \emptyset$ is idempotent/cyclic/symmetric if all operations in $\supp(\varpi)$ are such.
\end{definition}

It is well known and easy to see that all polymorphisms and fractional polymorphisms of a rigid core are idempotent.

%\begin{lemma}\label{lem:wei-fra}
%Any rigid core language $\Gamma$  has a cyclic weighted polymorphism if and only it has a cyclic fractional polymorphism.
%\end{lemma}
%
%\begin{proof}
%TODO. IN APPENDIX?
%\end{proof}

%Some results in~\cite{Kozik15:algebraic} that we describe below use Taylor operation, and specifically
%the condition that, for a core $\Gamma$ on $D$, $\fPolplus\Gamma$ contains a Taylor operation. However, the reader does not need to
%know what Taylor operations are because each cyclic operation is a Taylor operation and, by~\cite{Barto12:absorbing}, $\fPolplus\Gamma$ contains a %Taylor operation if and only if $\fPolplus\Gamma$ contains a cyclic operation of arity at least 2 (moreover, the arity can be chosen to be any prime %number $p>|D|$).
%Thus we will state results from ~\cite{Kozik15:algebraic} in terms of cyclic operations.
The following lemma is contained in the proof of Theorem 50 in~\cite{Kozik15:algebraic}.

\begin{lemma}\label{lem:trans}
Let $\Gamma$ be a rigid core on a set $D$. Then the following are equivalent:
\begin{enumerate}
\item $\fPolplus\Gamma$ contains a Taylor operation of arity at least 2;
%\item $\fPolplus\Gamma$ contains a cyclic operation of arity at least 2,
\item $\Gamma$ has a cyclic fractional polymorphism of (some) arity at least 2;
\item $\Gamma$ has a cyclic fractional polymorphism of every prime arity $p>|D|$.
\end{enumerate}
\end{lemma}

The following theorem is Corollary 51 from~\cite{Kozik15:algebraic}.

\begin{theorem}[\cite{Kozik15:algebraic}]\label{thm:VCSPhard}
Let $\Gamma$ be a valued constraint language that is a rigid core. If $\fPolplus\Gamma$ does not contain a Taylor operation then $\VCSP\Gamma$ is NP-hard.
\end{theorem}

Kozik and Ochremiak state a conjecture (which they attribute to L.~Barto) that the above theorem describes all NP-hard valued constraint languages, and all other languages
are tractable. Using Lemma~\ref{lem:trans}, we restate the original conjecture via cyclic fractional polymorphisms.

\begin{conjecture}[\cite{Kozik15:ICALP}]\label{con:VCSPtractable}
Let $\Gamma$ be a valued constraint language that is a rigid core. If $\Gamma$ has a cyclic fractional polymorphism of arity at least 2, then $\VCSP\Gamma$ is tractable.
\end{conjecture}

Note that, for a finite core $\Gamma$ (but with fixed $D$), the above condition can be checked in polynomial time. Indeed,
if $p>|D|$ is some fixed prime number, then it is sufficient to check for a cyclic fractional polymorphism of arity $p$.
Such polymorphisms, by definition, are solutions to a system of linear inequalities. Since the number of cyclic operations of arity $p$ on $D$ is constant, the system will have size polynomial in $\Gamma$ and its feasibility can be decided by linear programming.

% inserted "only" so that the text doesn't go over the right margin - vnk
For the case when (possibly infinite) $\Gamma$ consists only of $\{0,\infty\}$-valued functions, $\VCSP\Gamma$ is actually a CSP.
For such $\Gamma$, \emph{any} probability distribution on polymorphisms (of the same arity) is a fractional polymorphism.
Then a theorem and a conjecture (the latter now known as the Algebraic CSP Dichotomy Conjecture) equivalent to Theorem~\ref{thm:VCSPhard} and Conjecture~\ref{con:VCSPtractable} were given in~\cite{bulatov05:classifying}. One of several equivalent forms of the Algebraic CSP Dichotomy Conjecture is as follows.

\begin{conjecture}[\cite{bulatov05:classifying,Barto12:absorbing}]
\label{con:AlgCSPDich}
Let $\Gamma$ be a valued constraint language that is a rigid core and that consists of $\{0,\infty\}$-valued functions.
If $\Gamma$ has a cyclic polymorphism of arity at least 2, then $\VCSP\Gamma$ is tractable. Otherwise, $\VCSP\Gamma$ is NP-hard.
\end{conjecture}

In view of this, it is natural to call Conjecture~\ref{con:VCSPtractable} the {\em Algebraic VCSP Dichotomy Conjecture}.

\subsection{Basic LP relaxation}

Symmetric operations are known to be closely related to LP-based algorithms for CSP-related problems.
One algorithm in particular has been known to solve many VCSPs to optimality. This algorithm is based
on the so-called {\em basic LP relaxation}, or BLP, defined as follows.

Let $\mathbb M_n=\{\mu\ge 0\:|\:\sum_{x\in D^n}\mu(x)=1\}$ be the set of probability distributions over
labelings in $D^n$.
We also denote $\Delta=\mathbb M_1$; thus, $\Delta$ is the standard ($|D|-1$)-dimensional simplex.
The corners of $\Delta$ can be identified with elements in $D$.
For a distribution $\mu\in\mathbb M_n$ and a variable $v\in\{1,\ldots,n\}$, let
 $\mu_{[v]}\in \Delta$ be the marginal probability of distribution $\mu$ for $v$:
\begin{equation*}
\mu_{[v]}(a)\ = \sum_{x\in D^n:x_v=a} \mu(x) \qquad \forall a \in D.
\end{equation*}
Given a VCSP instance $\calI$ in the form~\eqref{eq:VCSPinst}, we define the value $\BLP(\calI)$ as follows:
%Clearly, the problem of minimizing $\hat f_\calI(\alpha)$ can be formulated as an LP of polynomial size (assuming that $\Gamma$ is finite).
\begin{eqnarray}
&& \hspace{-85pt} \BLP(\calI)\ =\ \min\ \sum_{t\in T}\sum_{x\in \dom f_t}\mu_t(x)f_t(x)  \label{eq:BLP} \\
 \mbox{s.t.~~} (\mu_t)_{[k]}&=&\alpha_{v(t,k)} \hspace{20pt} \forall t\in T,k\in\{1,\ldots,n_t\} \nonumber \\
\mu_t&\in&\mathbb M_{n_t}                      \hspace{29pt} \forall t\in T \nonumber \\
\mu_t(x)&=&0                                   \hspace{43pt} \forall t\in T,x\notin \dom f_t \nonumber \\
\alpha_v&\in&\Delta                            \hspace{40pt} \forall v\in V \nonumber
\end{eqnarray}
If there are no feasible solutions then $\BLP(\calI)=\infty$.
The objective function and all constraints in this system are linear, therefore this is a linear program. Its size is polynomial in the size of $\calI$, so
$\BLP(\calI)$ can be found in time polynomial in $|\calI|$.
%We call it the \emph{basic LP relaxation} of $\calI$ (BLP).

We say that BLP \emph{solves} $\calI$ if
$\BLP(\calI)=\min_{x\in D^n}f_\calI(x)$, and BLP
solves $\VCSP\Gamma$ if it solves all instances $\calI$ of $\VCSP\Gamma$.
%The BLP is said to solve problem $\VCSP\Gamma$ if the value of the relaxation is equal to the optimal value for every instance of $\VCSP\Gamma$.
 If BLP solves $\VCSP\Gamma$ and $\Gamma$ is a rigid core, then the optimal solution for every instance can be found
 by using the standard self-reducibility method. In this method, one goes through the variables in some order, finding
 $d\in D$ for the current variable $v$ such that instances $\calI$ and $\calI+u_d(v)$ have the same optimal value (which can be checked by BLP),
 updating $\calI:=\calI+u_d(v)$, and moving to the next variable. At the end, the instance will have a unique feasible assignment whose value
 is the optimum of the original instance.
 Note that in this case $\VCSP\Gamma$ is globally tractable.

\begin{theorem}[\cite{kolmogorov15:power}]
\label{thm:BLP}
BLP solves $\VCSP\Gamma$ if and only if, for every $m>1$, $\Gamma$ has a symmetric fractional polymorphism of arity $m$.
\end{theorem}

\begin{theorem}[\cite{kolmogorov15:power,tz13:stoc}]
\label{thm:finite-valued}
Let $\Gamma$ be a rigid core constraint language that is finite-valued.
If $\Gamma$ has a symmetric fractional polymorphism of arity $2$ then BLP solves $\VCSP\Gamma$, and so $\VCSP\Gamma$ is tractable. Otherwise, $\VCSP\Gamma$ is NP-hard.
\end{theorem}

%%%%%%%%%%%%%%%%%%%%%%%%%%%%%%%%%%%%%%%%%%%%%%%%%%%%%%%%%%%%%%%%%%%%%%%%%%%%%%%%%%%%%%%%%%%%%%%%%%%%%%%%%%%%%%
%%%%%%%%%%%%%%%%%%%%%%%%%%%%%%%%%%%%%%%%%%%%%%%%%%%%%%%%%%%%%%%%%%%%%%%%%%%%%%%%%%%%%%%%%%%%%%%%%%%%%%%%%%%%%%
%%%%%%%%%%%%%%%%%%%%%%%%%%%%%%%%%%%%%%%%%%%%%%%%%%%%%%%%%%%%%%%%%%%%%%%%%%%%%%%%%%%%%%%%%%%%%%%%%%%%%%%%%%%%%%
%%%%%%%%%%%%%%%%%%%%%%%%%%%%%%%%%%%%%%%%%%%%%%%%%%%%%%%%%%%%%%%%%%%%%%%%%%%%%%%%%%%%%%%%%%%%%%%%%%%%%%%%%%%%%%
%%%%%%%%%%%%%%%%%%%%%%%%%%%%%%%%%%%%%%%%%%%%%%%%%%%%%%%%%%%%%%%%%%%%%%%%%%%%%%%%%%%%%%%%%%%%%%%%%%%%%%%%%%%%%%
%%%%%%%%%%%%%%%%%%%%%%%%%%%%%%%%%%%%%%%%%%%%%%%%%%%%%%%%%%%%%%%%%%%%%%%%%%%%%%%%%%%%%%%%%%%%%%%%%%%%%%%%%%%%%%
%%%%%%%%%%%%%%%%%%%%%%%%%%%%%%%%%%%%%%%%%%%%%%%%%%%%%%%%%%%%%%%%%%%%%%%%%%%%%%%%%%%%%%%%%%%%%%%%%%%%%%%%%%%%%%
%%%%%%%%%%%%%%%%%%%%%%%%%%%%%%%%%%%%%%%%%%%%%%%%%%%%%%%%%%%%%%%%%%%%%%%%%%%%%%%%%%%%%%%%%%%%%%%%%%%%%%%%%%%%%%
%%%%%%%%%%%%%%%%%%%%%%%%%%%%%%%%%%%%%%%%%%%%%%%%%%%%%%%%%%%%%%%%%%%%%%%%%%%%%%%%%%%%%%%%%%%%%%%%%%%%%%%%%%%%%%

\section{Main Result}

\begin{definition}\label{feas-inst}
Let $\calI$ be a VCSP instance over variables $V$ with domain $D$.
The {\em feasibility instance}, $\Feas{\calI}$, associated to $\calI$ is a CSP instance obtained
from $\calI$ by replacing each constraint function $f_t$ with $\dom f_t$.
\end{definition}

For a language $\Gamma$, let $\Feas\Gamma=\{\dom f\mid f\in \Gamma\}$. Then the instances of the problem
$\CSP{\Feas\Gamma}$ are the instances $\Feas\calI$ where $\calI$ runs through all instances of $\VCSP\Gamma$.

\begin{definition}\label{1-infty-min-inst}
Let $\calI$ be a VCSP instance over variables $V$ with domain $D$.
For each variable $v\in V$, let $D_v=\{d\in D\mid d=\sigma(v) \mbox{ for some feasible solution } \sigma \mbox{ for } \calI\}$.
Then \emph{$(1,\infty)$-minimal instance} $\bar{\calI}$ associated with $\calI$ is the VCSP instance obtained from
$\calI$ by adding, for each $v\in V$, the constraint $u_{D_v}(x_v)$.
\end{definition}

%It is easy to see that if $u_{D_v}$ is as defined above then $u_{D_v}\in\langle\Feas\Gamma\rangle$.
%Since there are finitely many such functions, we can use Lemma~\ref{lem:exppower} and assume without loss of generality
%that all such functions (obtained from instances of $\VCSP\Gamma$) are contained in $\Gamma$.
%In this case, for any instance $\calI$ of $\VCSP\Gamma$, $\bar{\calI}$ is also an instance of this problem.

Note that if $\Gamma$ is a rigid core and the problem $\CSP{\Feas\Gamma}$ is tractable, then, for any instance $\calI$ of $\VCSP\Gamma$,
one can construct the associated $(1,\infty)$-minimal instance in polynomial time.
Indeed, to find out whether a given $d\in D$ is in $D_v$, one only needs to decide whether the CSP instance obtained from
$\Feas{\calI}$ by adding the constraint $u_{d}(x_v)$ is satisfiable. Since $\Gamma$ is a rigid core, the latter instance is also an instance
of $\CSP{\Feas\Gamma}$.

If $\Gamma$ is a rigid core then, for $\VCSP\Gamma$ to be tractable, $\Gamma$ must satisfy the assumption
of Conjecture~\ref{con:VCSPtractable}, and also, clearly, the feasibility part of the problem, $\CSP{\Feas\Gamma}$, must be tractable.
Our main result shows that if these necessary conditions are satisfied then $\VCSP\Gamma$ is indeed tractable.

\begin{theorem}\label{thm:main}
Let $\Gamma$ be a valued constraint language over domain $D$ that is a rigid core.
If the following conditions hold then $\VCSP{\Gamma}$ is tractable:
\begin{enumerate}
%\item $\Gamma_0$ is a rigid core;
\item $\Gamma$ has a cyclic fractional polymorphism of arity at least 2, and
\item $\CSP{\Feas{\Gamma}}$ is tractable.
\end{enumerate}
%If $\calI_0$ is an arbitrary instance of $\VCSP{\Gamma_0}$ and $\calI_0'$ its associated $(1,\infty)$-minimal instance,
%then $\Opt(\calI_0)=\BLP(\calI'_0)$.
%In particular, $\VCSP{\Gamma_0}$ is tractable. Otherwise
Otherwise, $\VCSP{\Gamma}$ is not tractable.
\end{theorem}

In Theorem~\ref{thm:main}, the intractability part for (absence of) the first condition follows from Theorem~\ref{thm:VCSPhard},
and it is obvious for the second condition. The tractability part follows from Theorem~\ref{thm:main-alg} below.

\begin{theorem}\label{thm:main-alg}
Let $\Gamma$ be an arbitrary language that has a cyclic fractional polymorphism of arity at least 2.
If $\calI$ is an instance of $\VCSP{\Gamma}$ and $\bar{\calI}$ is its associated $(1,\infty)$-minimal instance,
then $\Opt(\calI)=\BLP(\bar{\calI})$.
\end{theorem}

\noindent
 Indeed, if $\Gamma$ is a rigid core satisfying conditions (1) and (2) from Theorem~\ref{thm:main} and $\calI$ is an instance of $\VCSP\Gamma$ then the equality $\Opt(\calI)=\BLP(\bar{\calI})$ means that we can efficiently find the optimum value for $\calI$ by constructing $\bar{\calI}$ (which we can do efficiently because $\Gamma$ is a rigid core and $\CSP{\Feas{\Gamma}}$ is tractable) and then applying BLP to $\bar{\calI}$. Then we can find an optimal assignment for $\calI$ by self-reduction (see the discussion before Theorem~\ref{thm:BLP}).

Recall the notion of global tractability from Section~\ref{sec:VCSPdef}. The algorithm that we just described gives the following.

\begin{corollary}\label{cor:global}
Let $\Gamma$ be a valued constraint language over domain $D$ that is a rigid core.
If
\begin{enumerate}
%\item $\Gamma_0$ is a rigid core;
\item $\Gamma$ has a cyclic fractional polymorphism of arity at least 2, and
\item $\CSP{\Feas{\Gamma}}$ is globally tractable,
\end{enumerate}
 then $\VCSP{\Gamma}$ is globally tractable.
\end{corollary}

It also follows from Theorem \ref{thm:main-alg} that, for every language $\Gamma$ that has a cyclic fractional polymorphism of arity at least 2, $\VCSP{\Gamma}$ is polynomial time equivalent to $\CSP{\Feas{\Gamma}}$. In particular, \emph{any} complexity  classification of CSPs, whether it is the dichotomy as predicted by Conjecture \ref{con:AlgCSPDich} or anything else, gives a complexity classification of VCSPs.

Let us now discuss how Theorem~\ref{thm:main} can be combined with known CSP complexity classifications to obtain new, previously unknown, VCSP classifications which are tighter than Theorem~\ref{thm:main}.

As we explained in Section~\ref{sec:intro}, if the Algebraic CSP Dichotomy Conjecture holds, then condition (2) in Theorem~\ref{thm:main} can be omitted and all intractable VCSPs are NP-hard.
Since this conjecture holds when $|D|\le 3$~\cite{bulatov06:3-elementjacm,Schaefer78:complexity} or when $D$ is arbitrary finite, but $\Gamma$ contains all unary crisp functions~\cite{barto11:lics,Bulatov11:conservative}, we get the following corollaries.

\begin{corollary}\label{cor:VCSP-3-dichotomy}
Let $|D|\le 3$ and let $\Gamma$ be a valued constraint language that is a rigid core on $D$. If $\Gamma$ has a cyclic fractional polymorphism then the problem $\VCSP\Gamma$ is tractable, otherwise it is NP-hard.
\end{corollary}

For the case $|D|=2$, the tractable cases can be characterised by six specific cyclic fractional polymorphisms~\cite{cohen06:complexitysoft}, and it was shown in~\cite{Kozik15:algebraic} that the presence of any cyclic fractional polymorphism (when $|D|=2$) implies the presence of one of those six.
Also, Corollary~\ref{cor:VCSP-3-dichotomy} generalizes results from~\cite{Uppman13:icalp,Uppman14:stacs} where the dichotomy was shown for the special case when $|D|=3$ and all non-crisp functions in $\Gamma$ are unary. The specific conditions for tractability in~\cite{Uppman13:icalp,Uppman14:stacs} have not been shown to be directly implied by the presence of a cyclic fractional polymorphism, though.

\begin{corollary}\label{cor:VCSP-cons-dichotomy}
Let $\Gamma$ be a valued constraint language on $D$ that contains all unary crisp functions. If $\Gamma$ has a cyclic fractional polymorphism then the problem $\VCSP\Gamma$ is tractable, otherwise it is NP-hard.
\end{corollary}

Corollary~\ref{cor:VCSP-cons-dichotomy} generalizes a result from~\cite{Uppman13:icalp} where the dichotomy was shown for the special case when  $\Gamma$ includes all unary crisp functions and
all non-crisp functions in $\Gamma$ are unary. Again, the specific condition for tractability in~\cite{Uppman13:icalp} is not known to be directly implied by the presence of a cyclic fractional polymorphism.

It is shown in~\cite{Kozik15:algebraic} how Theorem~\ref{thm:main} implies the dichotomy results (including specific conditions for tractability)
for the finite-valued case from~\cite{tz13:stoc} (Theorem~\ref{thm:finite-valued}) and for the case when $\Gamma$ contains all unary functions taking values in $\{0,1\}$~\cite{Kolmogorov13:conservative}. The algorithm for the tractable case in~\cite{Kolmogorov13:conservative}
is somewhat similar in spirit to our algorithm, and actually inspired the latter.

Let us now explain how Theorem~\ref{thm:main} implies the tractability result from~\cite{Thapper15:Sherali} (stated below).
% changed "with" to "satisfying" so that the text doesn't go over the right margin - vnk
An idempotent operation $g\in \calO_D$ of arity at least 2 satisfying $g(y,x,x\ldots,x,x)=g(x,y,x,\ldots,x,x)=\ldots =g(x,x,x,\ldots,x,y)$ for all $x,y\in D$ is called a \emph{weak near-unanimity} operation.
The tractability result result from~\cite{Thapper15:Sherali} states that if $\fPolplus\Gamma$ contains weak near-unanimity operations of all but finitely many arities, then $\VCSP\Gamma$ is tractable (in fact, via a specific algorithm based on Sherali-Adams hierarchy, which does not follow from our results). This condition on $\fPolplus\Gamma$ is well known in the algebraic
approach to the CSP, it characterizes (when appropriately formulated) CSPs of bounded width~\cite{Barto14:jacm}.
So assume that $\fPolplus\Gamma$ satisfies this condition.
 Since $\fPolplus\Gamma\subseteq \Pol\Gamma$, the set $\Pol\Gamma$ also contains these operations, so $\CSP{\Feas{\Gamma}}$ is tractable by~\cite{Barto14:jacm}. Moreover, by~\cite{Barto12:absorbing},
$\fPolplus\Gamma$ then also contains a cyclic operation of arity at least 2. Now (the proof of) Theorem 50 of~\cite{Kozik15:algebraic} implies that $\Gamma$ has a cyclic fractional polymorphism of arity at least 2, and then tractability of $\VCSP\Gamma$ follows from Theorem~\ref{thm:main}.

We remark that some known VCSP classifications with tighter and more explicit characterisations of tractability can be easily derived from our main result, e.g. the classification for the Boolean case ($|D|=2$)  can be easily derived following
the lines of Section 8 of~\cite{Cohen13:algebraic}. However, it might take additional effort to derive some others - for example, the dichotomy result from~\cite{Thapper16:powerSA} was proved without using our theorem, and it is not known how to derive it from our main result.

\section{Proof of  Theorem~\ref{thm:main-alg}: Reduction to a block-finite language}
We will prove Theorem~\ref{thm:main-alg} by constructing, from a given (feasible) instance $\calI$, a finite valued constraint language $\Gamma'$ on some finite set $D'$ and an instance $\calI'$  of $\VCSP{\Gamma'}$ such that $\Opt(\calI)=\Opt(\bar{\calI})=\Opt(\calI')=\BLP(\calI')=\BLP(\bar{\calI})$. %and (b) BLP solves $\VCSP{\Gamma'}$.
The first equality is immediate from the definition $\bar{\calI}$, the second one will follow trivially from the construction of $\Gamma'$ and $\calI'$, and the last equality holds by Lemma~\ref{lem:BLPeq} below, while the key equality $\Opt(\calI')=\BLP(\calI')$ will follow from the fact that BLP solves $\VCSP{\Gamma'}$ that we prove, using Theorem~\ref{thm:BLP}, in Theorem~\ref{th:symm}. The construction is inspired by~\cite{Kolmogorov2015:hybrid}, where a similar technique of ``lifting'' a language was used in a different context.

Let $V$ be the set of variables of instance $\calI$, and let
\begin{equation}
f_{\calI}(x)\ =\ \sum_{t\in T} f_t(x_{v(t,1)},\ldots,x_{v(t,n_t)})\qquad \forall x:V\rightarrow D
\end{equation}
be its objective function. %(We have $v(t,i)\in V$ for $t\in T$, $i\in[n_t]$).
For the $(1,\infty)$-minimal instance $\bar{\calI}$, the objective function is
\begin{equation}
f_{\bar{\calI}}(x)\ =\ \sum_{t\in T} f_t(x_{v(t,1)},\ldots,x_{v(t,n_t)})+\sum_{v\in V}{u_{D_v}(x_v)} \qquad \forall x:V\rightarrow D
\end{equation}

Now let $D'_v=\{(v,a)\:|\:a\in D_v\}$ be a unique copy of $D_v$.
We now define a new language $\Gamma'$ over domain $D'=\bigcup_{v\in V}D'_v$ as follows:
$$
\Gamma'=\bigcup_{t\in T}\left\{f_t^{\langle v(t,1),\ldots,v(t,n_t)\rangle}, \dom f_t^{\langle v(t,1),\ldots,v(t,n_t)\rangle}\right\}
\cup\bigcup_{v\in V}\left\{u_{D'_v}\right\}\cup \{=_{D'}\}
$$
where functions $u_{D'_v}$ are as defined above, $=_{D'}$ is the binary $\{0,\infty\}$-valued function corresponding to the equality relation, and, for an $n$-ary function $f$ over $D$ and variables $v_1,\ldots,v_n\in V$,
we define function $f^{\langle v_1,\ldots,v_n\rangle}:(D')^n\rightarrow \Qc$ as follows:
$$
f^{\langle v_1,\ldots,v_n\rangle}(x)=\begin{cases}
f(\hat x) & \mbox{if }x=((v_1,\hat x_1),\ldots,(v_n,\hat x_n))  \\
\infty & \mbox{otherwise}
\end{cases}\qquad\forall x\in (D')^n
$$

The above mentioned instance $\calI'$ of $\VCSP{\Gamma'}$  is obtained from $\bar{\calI}$ by replacing each function $f_t$ with $f_t^{\langle v(t,1),\ldots,v(t,n_t)\rangle}$ and replacing each function $u_{D_v}$ with $u_{D'_v}$.

%\begin{lemma}
%We have $\Opt(\calI_0)=\Opt(\bar{\calI}_0)=\Opt(\calI)$.
%\end{lemma}
%\begin{proof}
%Obvious.
%\end{proof}

%There is a straightforward one-to-one correspondence between the sets of feasible solutions to $\BLP(\calI)=\BLP(\bar{\calI}_0)$, and this correspondence
%also preserves the values of solutions.

It is straightforward to check that there is a one-to-one correspondence between the sets of feasible solutions to BLP relaxations for $\calI'$ and $\bar{\calI}$, and that this correspondence also preserves the values of the solutions.

\begin{lemma}\label{lem:BLPeq}
We have $\BLP(\calI')=\BLP(\bar{\calI})$.
\end{lemma}

\begin{lemma}
\label{lem:cyclic}
If $\Gamma$ has a cyclic fractional polymorphism of arity $m>1$ then $\Gamma'$ has the same property.
\end{lemma}

\begin{proof}
Let $\omega$ be a cyclic fractional polymorphism of $\Gamma$.
%It is easy to see that each function $u_{D^\circ_v}$
%also has this fractional polymorphism because it is
%the restriction on variable $v$ of the effective domain of function $f_{\calI_0}$ which belongs to $\langle\Gamma_0\rangle$.
Fix an arbitrary element $d'\in D'$. For each operation $g\in \supp(\omega)$, define the operation $g'$ on $D'$ as follows:
\[
g'(x_1,\ldots,x_m)=\begin{cases}
(v,g(\hat{x}_1,\ldots,\hat{x}_m)) &\mbox{if } x_1\!=\!(v,\hat{x}_1),\ldots,x_m\!=\!(v,\hat{x}_m) \mbox{ for some } v\!\in\! V \\
d' &\mbox{otherwise }
\end{cases}
\]
Clearly, each operation $g'$ is cyclic.
Consider the fractional operation $\omega'$ on $D'$ such that $\omega(g')=\omega(g)$ for all $g\in supp(\omega)$.
It is straightforward to check that $\omega'$ is a fractional polymorphism of $\Gamma'$.
\end{proof}

To prove Theorem~\ref{thm:main-alg}, it remains to show that $\Opt(\calI')=\BLP(\calI')$. We will prove the more general fact that BLP solves $\VCSP{\Gamma'}$.
The properties of the language $\Gamma'$ that we use for this (apart from having a cyclic fractional polymorphism)
are given below in Definition~\ref{def:block-finite}.

%%%%%%%%%%%%%%%%%%%%%%%%%%%%%%%%%%%%%%%%%%%%%%%%%%%%%%%%%%%%%%%%%%%%%%%%%%%%%%%%%%%%%%%%%%%%%%%%%%%%%%%%%%%%%%
%%%%%%%%%%%%%%%%%%%%%%%%%%%%%%%%%%%%%%%%%%%%%%%%%%%%%%%%%%%%%%%%%%%%%%%%%%%%%%%%%%%%%%%%%%%%%%%%%%%%%%%%%%%%%%
%%%%%%%%%%%%%%%%%%%%%%%%%%%%%%%%%%%%%%%%%%%%%%%%%%%%%%%%%%%%%%%%%%%%%%%%%%%%%%%%%%%%%%%%%%%%%%%%%%%%%%%%%%%%%%
%%%%%%%%%%%%%%%%%%%%%%%%%%%%%%%%%%%%%%%%%%%%%%%%%%%%%%%%%%%%%%%%%%%%%%%%%%%%%%%%%%%%%%%%%%%%%%%%%%%%%%%%%%%%%%
%%%%%%%%%%%%%%%%%%%%%%%%%%%%%%%%%%%%%%%%%%%%%%%%%%%%%%%%%%%%%%%%%%%%%%%%%%%%%%%%%%%%%%%%%%%%%%%%%%%%%%%%%%%%%%
%%%%%%%%%%%%%%%%%%%%%%%%%%%%%%%%%%%%%%%%%%%%%%%%%%%%%%%%%%%%%%%%%%%%%%%%%%%%%%%%%%%%%%%%%%%%%%%%%%%%%%%%%%%%%%
%%%%%%%%%%%%%%%%%%%%%%%%%%%%%%%%%%%%%%%%%%%%%%%%%%%%%%%%%%%%%%%%%%%%%%%%%%%%%%%%%%%%%%%%%%%%%%%%%%%%%%%%%%%%%%
%%%%%%%%%%%%%%%%%%%%%%%%%%%%%%%%%%%%%%%%%%%%%%%%%%%%%%%%%%%%%%%%%%%%%%%%%%%%%%%%%%%%%%%%%%%%%%%%%%%%%%%%%%%%%%
%%%%%%%%%%%%%%%%%%%%%%%%%%%%%%%%%%%%%%%%%%%%%%%%%%%%%%%%%%%%%%%%%%%%%%%%%%%%%%%%%%%%%%%%%%%%%%%%%%%%%%%%%%%%%%

%\section{Block-finite languages}

\begin{definition}
A finite language $\Gamma$ %with the countable set $\Gamma$
is called {\em block-finite} if its domain $D$ can be
partitioned into disjoint subsets $\{D_v\:|\:v\in V\}$ such that
\begin{itemize}
\item[(a)] For any $a\in D_v$ with $v\in V$ there exists a polymorphism $g_a\in\calO^{(1)}$ of $\Feas\Gamma$
such that $g_a(b)=a$ for all $b\in D_v$.
\item[(b)] For any $n$-ary function $f\in\Gamma$,
the relation $\dom f$ (viewed as a function $D^n\rightarrow\{0,\infty\}$) belongs to $\Gamma$.
Furthermore, the binary equality relation on $D$, denoted as $=_D:D^2\rightarrow\{0,\infty\}$, also belongs to $\Gamma$.
\item[(c)] Any $n$-ary function $f\in\Gamma-\{=_D\}$ satisfies $\dom f\subseteq D_{v_1}\times\ldots\times D_{v_n}$
for some $v_1,\ldots,v_n\in V$.
\end{itemize}
\label{def:block-finite}
\end{definition}

It is easy to see that the language $\Gamma'$ defined in the previous section is block-finite.
It obviously has properties (b) and (c), and it has property (a) because the instance
$\calI'$ is $(1,\infty)$-minimal. Indeed, if $a=(v,d)\in D'_v$ then, by definition, $\calI$
has a feasible solution $\sigma:V\rightarrow D$ with $\sigma(v)=d$.
Define function $g_a$ as follows: for each $a'=(v',d')\in D'$, set $g_a(a')=(v',\sigma(v'))$.
It is easy to check that $g_a$ has the required properties.

From now on, we forget about the original language $\Gamma$ from the previous section and about the specific language $\Gamma'$ and
work with an arbitrary block-finite language that has a cyclic fractional polymorphism of arity at least 2. For simplicity, we denote our language
by $\Gamma$.
Note that $\Gamma$ is not necessarily a (rigid) core, but this property is not required in Theorem~\ref{thm:BLP}.
By Theorem~\ref{thm:BLP}, in order to prove Theorem~\ref{thm:main-alg}, it remains to show the following.

\begin{theorem}
Suppose that a block-finite language $\Gamma$
admits a cyclic fractional polymorphism $\nu$ of arity at least 2.
Then, for every $m\ge 2$, $\Gamma$ admits a symmetric fractional polymorphism $\omega^{\tt sym}_m$ of arity $m$.
%\label{th:sym2}
%\end{theorem}
%
%\begin{theorem}
%Suppose that a block-finite language $\Gamma$
%admits a symmetric fractional polymorphism $\nu$ of arity  $m-1\ge 2$.
%Then $\Gamma$ admits  symmetric fractional polymorphism of arity $m$.
\label{th:symm}
\end{theorem}

In the rest of the paper we prove Theorem~\ref{th:symm}.
This will be done in two steps: (i) using the existence of $\nu$, prove the existence of $\omega^{\tt sym}_2$;
(ii) using the existence of $\omega^{\tt sym}_{m-1}$ for some $m\ge 3$, prove the existence of $\omega^{\tt sym}_{m}$.
The claim will then follow by induction on $m$.

Note that for finite-valued languages step (i) was proved in~\cite{tz13:stoc}
(or rather a very closely related statement), while step (ii) was established in~\cite{kolmogorov15:power}.
However, in both cases it was essential that the language is finite-valued.
The arguments in~\cite{kolmogorov15:power,tz13:stoc} seem to break down when infinities are allowed.
 For example, we were unable to extend the approach in~\cite{tz13:stoc} that exploits the connectivity
 of a certain graph on $D$. To deal with block-finite languages, we will introduce (in Section~\ref{sec:spec}) a new technical tool
 where we first prove, via Farkas Lemma, the existence of a certain function with special properties in $\closure\Gamma$.

\section{A graph of generalized operations}\label{sec:G}
In this section we describe a basic tool that will be used for constructing new fractional polymorphisms,
namely a graph of generalized operations introduced in~\cite{kolmogorov15:power}.

Let $\calO^{(m\rightarrow m)}$ be the set of mappings ${\bf g}:D^m\rightarrow D^m$ and let ${\mathds 1}\in\calO^{(m\rightarrow m)}$ be the identity mapping.
Consider a sequence $x$ of $m$ labelings $x\in [D^n]^m$; this means that $x=(x^1,\ldots,x^m)$ where $x^i\in D^n$ (think of $x$ as an $m \times n$ matrix whose rows are $x^1$, \ldots, $x^m$).
For an $n$-ary function $f$, we define $f^m(x)=\frac{1}{m}(f(x^1)+\ldots+f(x^m))$ (thus $f^m(x)$ is the average value of $f$ on the rows of $x$).
For a mapping ${\bf g}=(g_1,\ldots,g_m)\in\calO^{(m\rightarrow m)}$, we also
 denote $x^{{\bf g}i}=g_i(x)$ for $i\in [m]$ and ${\bf g}(x)=(x^{{\bf g}1},\ldots,x^{{\bf g}m})$ (so $g(x)$ is an $m \times n$ matrix where row $i$ is obtained by column-wise application of $g_i$ to $x$).

A probability distribution $\rho$ over $\calO^{(m\rightarrow m)}$ will be called
a {\em (generalized) fractional polymorphism of $\Gamma$ of arity $m\rightarrow m$} if each function $f\in\Gamma$  satisfies
\begin{eqnarray}
\sum_{{\bf g}\in\supp(\rho)} \rho({\bf g})f^m({\bf g}(x))&\le& f^m(x)\qquad \forall x\in  [\dom f]^m \label{eq:genfpol}
\end{eqnarray}

We will sometimes represent fractional polymorphisms of arity $m$ and generalised fractional polymorphisms of arity $m\rightarrow m$ as vectors in $\mathbb{R}^{\calO^{(m)}}$ and $\mathbb{R}^{\calO^{(m\rightarrow m)}}$, respectively. For $g\in \calO^{(m)}$ and $\mathbf{g}\in \calO^{(m\rightarrow m)}$,
we denote the corresponding characteristic vectors by $\chi_{g}$ and $\chi_{{\bf g}}$ respectively.
It can be checked
that a generalized fractional polymorphism $\rho$ of arity $m\rightarrow m$ can be converted into a fractional polymorphism $\rho'$ of arity $m$, as follows:
\[
\rho'=\sum_{\mathbf{g} =(g_1,\ldots,g_m)\in\supp(\rho)}{\frac{\rho({\bf g})}{m}(\chi_{g_1}+\ldots+\chi_{g_m})}.
\]

%\begin{remark}
%It is known~\cite{Kolmogorov13:power} and easy to check that a generalized fractional polymorphism $\rho$ of arity $m\rightarrow m$ can be converted into a fractional polymorphism $\rho'$ of arity $m$, as follows:
%\[
%\rho'=\sum_{\mathbf{g} =(g_1,\ldots,g_m)\in\supp(\rho)}{\frac{\rho({\bf g})}{m}(\chi_{g_1}+\ldots+\chi_{g_m})}.
%\]
%\end{remark}

We will use the following construction in several parts of the proof. Assume that we have some probability distribution $\omega$ with a finite support such that (i)
each element $s\in\supp(\omega)$ corresponds to an element of $\calO^{(m\rightarrow m)}$ denoted as ${\mathds 1}^s$, and (ii)
this distribution satisfies the following property for each $f\in\Gamma$:
\begin{subequations}
\begin{eqnarray}
\sum_{s\in\supp(\omega)} \omega(s)f^m({\mathds 1}^s(x))&\le& f^m(x)\qquad \forall x\in  [\dom f]^m \label{eq:G:omega}
\end{eqnarray}
Condition~\eqref{eq:G:omega} then can be rephrased as saying that vector $\sum_{s\in\supp(\omega)}\omega(s)\chi_{{\mathds 1}^s}$ is a  fractional polymorphism of $\Gamma$ of arity $m\rightarrow m$.
We will also consider the following condition:
%In one of the results we will also assume that $\omega$ satisfies
\begin{eqnarray}
\sum_{s\in\supp(\omega)} \omega(s)f(x^{{\mathds 1}^si})&\le& f^{m-1}(x_{-i})\qquad \forall x\in   [\dom f]^m,i\in[m] \label{eq:G:omega'}
\end{eqnarray}
\end{subequations}
where $x_{-i}\in[\dom f]^{m-1}$ denotes the sequence of $m-1$ labelings obtained from $x$ by removing the $i$-th labeling.
Note that condition \eqref{eq:G:omega'} implies  \eqref{eq:G:omega}
(since summing \eqref{eq:G:omega'} over $i\in[m]$ and dividing by $m$ gives \eqref{eq:G:omega}).
The second condition will be used only in one of the results;
unless noted otherwise, $\omega$ is only assumed to satisfy~\eqref{eq:G:omega}.

For a mapping ${\bf g}\in \calO^{(m\rightarrow m)}$ denote ${\bf g}^s={\mathds 1}^s\circ{\bf g}$.
(This notation is consistent with the earlier one since ${\mathds 1}^s\circ{\mathds 1}={\mathds 1}^s$ for any $s$).
We use ${\bf g}^{s_1\ldots s_k}$ to denote $(\ldots({\bf g}^{s_1})^{\ldots})^{s_k}={\mathds 1}^{s_k}\circ\ldots\circ{\mathds 1}^{s_1}\circ{\bf g}$.
Next, define a directed graph $(\mathbb G,E)$ as follows:
\begin{itemize}
\item $\mathbb G=\{{\mathds 1}^{s_1\ldots s_k}\:|\:s_1,\ldots,s_k\in\supp(\omega),k\ge 0\}$
is the set of all mappings that can be obtained from ${\mathds 1}$ by applying operations from $\supp(\omega)$;
\item $E=\{({\bf g},{\bf g}^s)\:|\:{\bf g}\in \mathbb G,s\in\supp(\omega)\}$.
\end{itemize}
This graph can be decomposed into strongly connected components, yielding a directed acyclic
graph (DAG) on these components. We define
 ${\sf Sinks}({\mathbb G,E})$ to be the set of those strongly connected components $\mathbb H\subseteq \mathbb G$
of $(\mathbb G,E)$ that are sinks of this DAG (i.e.\ have no outgoing edges). Any DAG has at least one sink, therefore ${\sf Sinks}({\mathbb G,E})$ is non-empty. We denote
$\mathbb G^\ast=\bigcup_{\mathbb H\in {\sf Sinks}({\mathbb G,E})} \mathbb H\subseteq \mathbb G$
and $Range_n({\mathbb G}^\ast)=\{{\bf g}^\ast(x)\:|\:{\bf g}^\ast\in\mathbb G^\ast,x\in [D^n]^m\}$.
%It follows from this definition that $[Range({\mathbb G}^\ast)]^n=\{{\bf g}^\ast(x)\:|\:{\bf g}^\ast\in\mathbb G^\ast,x\in [D^n]^m\}$.
%
%and, for a function $f\in\Gamma$ of arity $n$,
%\begin{subequations}
%\begin{eqnarray}
%Range_n({\mathbb G}^\ast)\;\;\;\;\!&=&\{{\bf g}^\ast(x)\:|\:{\bf g}^\ast\in\mathbb G^\ast,x\in[D^n]^m\} \\
%Range_n({\mathbb G}^\ast;f)&=&\{{\bf g}^\ast(x)\:|\:{\bf g}^\ast\in\mathbb G^\ast,x\in[D^n]^m\cap[\dom f]^m\}
%\end{eqnarray}
%\end{subequations}
Also, for a tuple $\hat x\in D^m$
we will denote $\mathbb G(\hat x)=\{{\bf g}(\hat x)\:|\:{\bf g}\in\mathbb G\}\subseteq D^m$.

The following facts can easily be shown (see Appendices~\ref{sec:proof:prop:Range}, ~\ref{sec:proof:prop:Ghatx}).
%The following  observation will be useful.
\begin{proposition}
(a) If ${\bf g},{\bf h}\in\mathbb G$ then ${\bf h}\circ{\bf g}\in\mathbb G$.
%and there is a path in $(\mathbb G,E)$ from ${\bf g}$ to ${\bf h}\circ{\bf g}$.
Moreover, if ${\bf g}\in\mathbb H\in{\sf Sinks}({\mathbb G,E})$ then ${\bf h}\circ{\bf g}\in\mathbb H$. \\
(b) Consider connected components ${\mathbb H}',{\mathbb H}\in{\sf Sinks}({\mathbb G,E})$.
For each ${\bf g}'\in\mathbb H'$ there exists ${\bf g}\in\mathbb H$ satisfying ${\bf g}\circ{\bf g}'={\bf g}'$. \\
(c)
For each $x\in Range_n({\mathbb G}^\ast)$ and $\mathbb H\in{\sf Sinks}({\mathbb G,E})$
there exists ${\bf g}\in\mathbb H$ satisfying ${\bf g}(x)=x$.
\label{prop:Range}
\end{proposition}

\begin{proposition}
Suppose that $\hat x\in Range_1(\mathbb G^\ast)$ and $x\in\mathbb G(\hat x)$. \\
(a) There holds $x\in Range_1(\mathbb G^\ast)$. \\
(b) There exists ${\bf g}\in\mathbb G$ such that ${\bf g}(x)=\hat x$.
\label{prop:Ghatx}
\end{proposition}

We now state main theorems related to the graph $(\mathbb G,E)$, that are slight extensions of the results in \cite{kolmogorov15:power}. Their proofs use the same techniques as \cite{kolmogorov15:power} and can be found in Appendices~\ref{sec:proof:G:rho},~\ref{sec:proof:th:Ga},~\ref{sec:proof:th:Gb}.
\begin{theorem}
Let $\widehat{\mathbb G}$ be a subset of $\mathbb G$
satisfying the following property: for each ${\bf g}\in\mathbb G$
there exists a path in $(\mathbb G,E)$ from ${\bf g}$ to some node $\widehat{\bf g}\in\widehat{\mathbb G}$.
Then there exists a fractional polymorphism $\rho$ of $\Gamma$ of arity $m\rightarrow m$
with $\supp(\rho)=\widehat{\mathbb G}$.
\label{th:G:rho}
\end{theorem}
We will use this result either for the set $\widehat{\mathbb G}=\mathbb G$ or for the set $\widehat{\mathbb G}=\mathbb G^\ast$;
clearly, both choices satisfy the condition of the theorem. The first choice
gives that $\Gamma$ admits a fractional polymorphism $\rho$ with $\supp(\rho)=\mathbb G$;
therefore, if ${\bf g}\in\mathbb G$, $f\in\closure\Gamma$ and $x\in[\dom f]^m$ then ${\bf g}(x)\in[\dom f]^m$.
%In particular, we have $Range_n(\mathbb G^\ast;f)\subseteq[\dom f]^m$.

%In the next theorem we denote $Range_n({\mathbb G}^\ast)=\{{\bf g}^\ast(x)\:|\:{\bf g}\in\mathbb G^\ast,x\in
%In the next theorem we use the following nota

\begin{theorem} Consider function $f\!\in\!\closure\Gamma$ of arity $n$ and labelings $x\hspace{-2.1pt}\in\hspace{-2.1pt} Range_n(\mathbb G^\ast)\cap\! [\dom f]^m$.
\begin{itemize}
\item[(a)] There holds $f^m({\bf g}(x))=f^m(x)$ for any ${\bf g}\in\mathbb G$.

%\item[(b)] For any $\mathbb H\in {\sf Sinks}({\mathbb G,E})$ there exists a probability distribution $\lambda$ over $\mathbb H$ with the following property:
\item[(b)] Suppose that condition~\eqref{eq:G:omega'} holds.
Then there exists a probability distribution $\lambda$ over $\mathbb G^\ast$ (which is independent of $f,x$)
such that $f^\lambda_{i'}(x)=f^\lambda_{i''}(x)$ for any $i',i''\in[m]$ where
\begin{equation}
f^\lambda_i(x)=\sum_{{\bf g}\in\mathbb G^\ast}\lambda_{\bf g}f(x^{{\bf g}i}) %\qquad\forall y\in[\dom f]^m
\label{eq:flambdai}
\end{equation}
\end{itemize}
\label{th:G}
\end{theorem}

%\begin{proof}
%\myparagraph{Part (d)} We know that ${\bf g}(\hat x)=x$ for some ${\bf g}\in\mathbb G$.
%We also have ${\bf g}^\pi(\hat x)=x^\pi$ and ${\bf g}^\pi\in\mathbb G$. This implies the claim.
%\end{proof}

%%%%%%%%%%%%%%%%%%%%%%%%%%%%%%%%%%%%%%%%%%%%%%%%%%%%%%%%%%%%%%%%%%%%%%%%%%%%%%%%%%%%%%%%%%%%%%%%%%%%%%%%%%%%%
%%%%%%%%%%%%%%%%%%%%%%%%%%%%%%%%%%%%%%%%%%%%%%%%%%%%%%%%%%%%%%%%%%%%%%%%%%%%%%%%%%%%%%%%%%%%%%%%%%%%%%%%%%%%%
%%%%%%%%%%%%%%%%%%%%%%%%%%%%%%%%%%%%%%%%%%%%%%%%%%%%%%%%%%%%%%%%%%%%%%%%%%%%%%%%%%%%%%%%%%%%%%%%%%%%%%%%%%%%%
%%%%%%%%%%%%%%%%%%%%%%%%%%%%%%%%%%%%%%%%%%%%%%%%%%%%%%%%%%%%%%%%%%%%%%%%%%%%%%%%%%%%%%%%%%%%%%%%%%%%%%%%%%%%%
%%%%%%%%%%%%%%%%%%%%%%%%%%%%%%%%%%%%%%%%%%%%%%%%%%%%%%%%%%%%%%%%%%%%%%%%%%%%%%%%%%%%%%%%%%%%%%%%%%%%%%%%%%%%%
%%%%%%%%%%%%%%%%%%%%%%%%%%%%%%%%%%%%%%%%%%%%%%%%%%%%%%%%%%%%%%%%%%%%%%%%%%%%%%%%%%%%%%%%%%%%%%%%%%%%%%%%%%%%%
%%%%%%%%%%%%%%%%%%%%%%%%%%%%%%%%%%%%%%%%%%%%%%%%%%%%%%%%%%%%%%%%%%%%%%%%%%%%%%%%%%%%%%%%%%%%%%%%%%%%%%%%%%%%%
%%%%%%%%%%%%%%%%%%%%%%%%%%%%%%%%%%%%%%%%%%%%%%%%%%%%%%%%%%%%%%%%%%%%%%%%%%%%%%%%%%%%%%%%%%%%%%%%%%%%%%%%%%%%%

\section{Constructing special functions}\label{sec:spec}

In this section, we construct special functions in $\langle\Gamma\rangle$ that play an important role in the proof of Theorem~\ref{th:symm}.

For a sequence $x=(x^1,\ldots,x^m)\in D^m$ and a permutation $\pi$ of $[m]$, we define $x^\pi=(x^{\pi(1)},\ldots,x^{\pi(m)})$.
Similarly, for a mapping  ${\bf g}=(g_1,\ldots,g_m)\in\calO^{(m\rightarrow m)}$ define ${\bf g}^\pi=( g_{\pi(1)},\ldots, g_{\pi(m)})$.
Let $\Omega$ be the set of mappings ${\bf g}\in\calO^{(m\rightarrow m)}$ that satisfy the following condition:
%\begin{subequations}
\begin{eqnarray}
~&&\hspace{-65pt}\,~~\mbox{ ${\bf g}^\pi(x)={\bf g}(x^\pi)$ for any $x\in D^m$ and any permutation $\pi$ of $[m]$. } \label{eq:Omega:def:a} \\
~&&\hspace{-65pt}\,~~\mbox{ Equivalently, $g_{\pi(i)}(x)=g_i(x^\pi)$ for any $i\in[m]$.} \nonumber
%&&\hspace{-155pt}\bullet~\mbox{ If $f\in\Gamma$ and $x\in[\dom f]^m$ then ${\bf g}(x)\in[\dom f]^m$. } \label{eq:Omega:def:b}
\end{eqnarray}

\begin{proposition}\label{prop:Omegaclosed} If ${\bf g}, {\bf h} \in \Omega$, then ${\bf g} \circ {\bf h} \in \Omega$.
\end{proposition}
\begin{proof}
Just note that
$$\left( {\bf g} \circ {\bf h} \right)^\pi(x) = {\bf g}^\pi \left( {\bf h}(x) \right) = {\bf g} \left( {\bf h}^\pi(x) \right) = {\bf g} \left( {\bf h}(x^\pi) \right) = \left( {\bf g} \circ {\bf h} \right)(x^\pi)$$
for any $x \in D^m$.
\end{proof}
%\end{subequations}
%The condition  \eqref{eq:Omega:def:a} for ${\bf g}=(g_1,\ldots,g_m)$ and a permutation $\pi$ can be written equivalently as follows: $g_{\pi(i)}(x)=g_i(x^\pi)$ for any $i\in[m]$.

Consider all generalized fractional polymorphisms $\omega$ of $\Gamma$ of arity $m\rightarrow m$
satisfying $\supp(\omega)\subseteq\Omega$.
At least one such polymorphism exists, namely $\omega=\chi_{\mathds 1}$ where ${\mathds 1}\in\calO^{(m\rightarrow m)}$
is the identity mapping. Among such $\omega$'s, pick one with the largest support.
It exists due to the following observation: if $\omega',\omega''$ are generalized fractional polymorphisms
of $\Gamma$ of arity $m\rightarrow m$ then so is the vector $\omega=\frac{1}{2}[\omega'+\omega'']$,
and $\supp(\omega)=\supp(\omega')\cup\supp(\omega'')$.

Let us apply the construction of Section~\ref{sec:G} starting with the chosen distribution $\omega$,
where for ${\bf g}\in\supp(\omega)$ we define operation ${\mathds 1}^{\bf g}\in\calO^{(m\rightarrow m)}$ via ${\mathds 1}^{\bf g}={\bf g}$. Let the resulting graph be $(\mathbb G, E)$.
It is straightforward to check that condition~\eqref{eq:G:omega} holds: it simply expresses
the fact that $\omega$ is a generalized fractional polymorphsism of $\Gamma$ of arity $m\rightarrow m$.
%It can also be checked that
 %$\mathbb G=\supp(\omega)\subseteq \Omega$.

%\begin{proposition}
%If ${\bf g},{\bf h}\in\mathbb G$ then ${\bf h}\circ{\bf g}\in\mathbb G$ and $({\bf g},{\bf h}\circ{\bf g})\in E$.
%\label{prop:Gcomposition}
%\end{proposition}
%\begin{proof}
%TODO
%\end{proof}

\begin{proposition} \label{prop:GisOmega} It holds that $\supp(\omega) = \mathbb G$.
\end{proposition}
\begin{proof} If ${\bf g } \in \supp(\omega)$ then ${\bf g} = {\mathds 1}^{\bf g} \in \mathbb G$.
Conversely, suppose that ${\bf g}\in\mathbb G$.
%For the other direction, let us first show that $\mathbb G \subset \Omega$. Indeed, for ${\bf g} \in \mathbb G$,
We can write
${\bf g} = {\mathds 1}^{{\bf g}_k}\circ\ldots\circ{\mathds 1}^{{\bf g}_1}={\bf g}_k \circ \dots \circ {\bf g}_1$ with ${\bf g}_1, \dots, {\bf g}_k \in \supp(\omega) \subseteq \Omega$.
Since $\Omega$ is closed under composition by Proposition \ref{prop:Omegaclosed}, we get ${\bf g} \in \Omega$.
By Theorem \ref{th:G:rho} there exists a generalized fractional polymorphism $\rho$ with $\supp(\rho) = \mathbb G$, and so ${\bf g}\in\supp(\rho)$.
By maximality of $\omega$ we get ${\bf g}\in \supp(\omega)$.
\end{proof}

%For a tuple $\hat x\in D^m$ we denote $\mathbb G(\hat x)=\{{\bf g}(\hat x)\:|\:{\bf g}\in\mathbb G\}\subseteq D^m$.
%\begin{lemma}
%Suppose that $\hat x\in Range_1(\mathbb G^\ast)$. \\
%(a) If $x\in\mathbb G(\hat x)$ then $x\in Range_1(\mathbb G^\ast)$.  \\
%(b) For any $x\in\mathbb G(\hat x)$ there exists ${\bf g}\in\mathbb G$ with ${\bf g}(x)=\hat x$. \\
%(c) There holds $\hat x\in\mathbb G(\hat x)$. \\
%%(d) If $x\in\mathbb G(\hat x)$ and $\pi$ is a permutation of $[m]$ then $x^\pi\in\mathbb G(\hat x)$.
%\label{lemma:rangePermutationInv}
%\end{lemma}
%\begin{proof}
%TODO
%\myparagraph{Part (d)} We know that ${\bf g}(\hat x)=x$ for some ${\bf g}\in\mathbb G$.
%We also have ${\bf g}^\pi(\hat x)=x^\pi$ and ${\bf g}^\pi\in\mathbb G$. This implies the claim.
%\end{proof}
In the remainder of this section we prove the following theorem.
\begin{theorem}
For any $\hat x\in D^m$ there exists a function $f\in\closure\Gamma$ of arity $m$ with $\arg\min f=\mathbb G(\hat x)$.
%(a) For any $x\in P$ and permutation $\pi$ of $[m]$, $x^\pi\in\arg\min f$.
%(b) if $x\in\arg\min f$ then there exists ${\bf g}\in \mathbb G$
% then $\sigma(x)=m$, or equivalently $x_1=\ldots=x_m$.
\label{th:specialInstance}
\end{theorem}
%First, we make the following observation.

%Note that by Lemma TODO, we have $x^\pi\in P$ for any $x\in P$ and a permutation $\pi$ of $[m]$.
\begin{proof}
Let $\Gamma^+$ be the set of pairs $(f,x)$ with $f\in\Gamma$ and $x\in[\dom f]^m$.
Let $\Omega'\subseteq\Omega$ be the set of mappings ${\bf g}\in\Omega$ that
satisfy  ${\bf g}(x)\in[\dom f]^m$ for all $(f,x)\in\Gamma^+$. Note that $\mathbb G=\supp(\omega)\subseteq\Omega'$.
By the choice of $\omega$, the following system does not have a solution with rational $\rho$:
\begin{subequations}
\begin{eqnarray}
\rho({\bf g})&\ge &0 \qquad \forall {\bf g}\in \Omega' \\
\sum_{{\bf g}\in\Omega'}\rho({\bf g})f^m(x)-\sum_{{\bf g}\in\Omega'}\rho({\bf g})f^m({\bf g}(x)) &\ge&0\qquad \forall (f,x)\in\Gamma^+ \\
\sum_{{\bf g}\in\Omega'-\mathbb G}-\rho({\bf g})&<&0
\end{eqnarray}
\end{subequations}
Next, we use the following well-known result (see, e.g.~\cite{Schrijver86:book}).
\begin{lemma}[Farkas Lemma]
Let $A$ be a $p\times q$ matrix and $b$ be a $p$-dimensional vector. Then exactly one of the following is true:
\begin{itemize}
\item There exists $\lambda\in\mathbb R^q$ such that $A\lambda=b$ and $\lambda\ge 0$.
\item There exists $\mu\in\mathbb R^p$ such that $\mu^TA\ge 0$ and $\mu^T b<0$.
\end{itemize}
If $A$ and $b$ are rational then $\lambda$ and $\mu$ can also be chosen in $\mathbb Q^q$ and $\mathbb Q^p$, respectively.
\label{lemma:Farkas}
\end{lemma}
By this lemma, the following system has a solution with rational $\lambda\ge 0$:
\begin{subequations}\label{eq:GJALKJDGHAKSF}
\begin{eqnarray}
\lambda({\bf g})+\sum_{(f,x)\in\Gamma^+}\lambda(f,x)(f^m(x)-f^m({\bf g}(x)) &=&0\hspace{8pt}\qquad \forall {\bf g}\in\mathbb G \\
\lambda({\bf g})+\sum_{(f,x)\in\Gamma^+}\lambda(f,x)(f^m(x)-f^m({\bf g}(x)) &=&-1\qquad \forall {\bf g}\in\Omega'-\mathbb G
\end{eqnarray}
\end{subequations}
%All coefficients of this system are rational numbers, therefore there exists a solution $\lambda\ge 0$ which is
%rational. (This follows, for example, from the Fourier-Motzkin elimination procedure.)

We will now define several instances of $\VCSP\Gamma$ where it will be convenient to use constraints with rational positive weights;
these weights can always be made integer by multiplying the instances by an appropriate positive integer,
which would not affect the reasoning, but make notation cumbersome.

We will define a $\Gamma$-instance $\calI$ with $m|D|^m$ variables $\calV=\{(i,z)\:|\:i\in[m],z\in D^m\}$.
The labelings $\calV\rightarrow D$ for this instance can be identified with mappings ${\bf g}=(g_1,\ldots,g_m)\in\calO^{(m\rightarrow m)}$,
if we define ${\bf g}(i,z)=g_i(z)$ for the coordinate $(i,z)\in\calV$.
We define  the cost function of $\calI$ as follows:
\begin{equation}
f_\calI({\bf g})=\sum_{(f,x)\in\Gamma^+,\lambda(f,x)\ne 0}\lambda(f,x)f^m({\bf g}(x))
\qquad \forall{\bf g}\in\calO^{(m\rightarrow m)}
\end{equation}
%
%where $C$ is a positive integer such that $C\lambda(f,x)\in\mathbb Z$ for all $(f,x)\in\Gamma^+$.
%Note, we have $f_{\calI}\in\closure\Gamma$ since we can interpret expression $N\cdot f^m({\bf g}(x))$ for ${\bf g}=(g_1,\ldots,g_m)$ as repeating terms $[f(g_1(x)+\ldots+f(g_m(x))]$ ~
%$N/m$ times (or omitting them, if $N=0$).
%The value of labeling ${\bf g}\in\calO^{(m\rightarrow m)}$ on variable $(i,z)\in\calV$ will  be denoted as ${\bf g}(i,z)$,
%and is equal to ${\bf g}(i,z)=g_i(z)$ assuming that ${\bf g}=(g_1,\ldots,g_m)$.
%Note that for each $i\in[m]$ and $z\in D^m$ we have a variable $(i,z)$ ;
%the value of labeling ${\bf g}=(g_1,\ldots,g_m)$ on the $i$-th copy of variable $z$ equals $g_i(z)\in D$.
%We will also denote this value as ${\bf g}(i,z)$.
%Note, for any ${\bf g}\in\Omega'$ we have $f_\calI({\bf g})=Cm\sum_{(f,x)\in\Gamma^+}\lambda(f,x)f^m({\bf g}(x))<\infty$.
%
From~\eqref{eq:GJALKJDGHAKSF}  we get
\begin{subequations}
\label{eq:KDNFAIDFA}
\begin{eqnarray}
f_\calI({\mathds 1})&=&f_\calI({\bf g})-\lambda({\bf g})\;\le\; f_\calI({\bf g}) \;<\;\infty\qquad \forall {\bf g}\in\mathbb G \\
f_\calI({\mathds 1})&<&f_\calI({\bf g})-\lambda({\bf g})\;\le\; f_\calI({\bf g}) \;<\;\infty\qquad \forall {\bf g}\in\Omega'-\mathbb G
%f_\calI({\bf g})&=&f_\calI({\bf g}^\pi)\hspace{26pt}\qquad \forall {\bf g}\in\calO^{(m\rightarrow m)},\mbox{ permutation $\pi$ of $[m]$} \label{eq:instance:first':c}
%f_\calI({\bf g})&=&\infty\hspace{47pt}\qquad \forall {\bf g}\in\calO^{(m\rightarrow m)}-\Omega
\end{eqnarray}
Furthermore, $f^m(\cdot)$ is invariant with respect to permuting its arguments, and thus
\begin{eqnarray}
f_\calI({\bf g})&=&f_\calI({\bf g}^\pi)\hspace{26pt}\qquad \forall {\bf g}\in\calO^{(m\rightarrow m)},\mbox{ permutation $\pi$ of $[m]$} \label{eq:instance:first':c}
\end{eqnarray}

\end{subequations}
%(The last condition holds since $f^m(\cdot)$ is invariant with respect to permuting its arguments.)
%Furthermore, $f_\calI({\bf g})=f_\calI({\bf g}^\pi)$ for any ${\bf g}\in\mathbb G$ and any permutation $\pi$ of $[m]$.

Let $T$ be the set of tuples $(i,j,x,y)$ where $i,j\in[m]$, $x,y\in D^m$ and $i=\pi(j)$, $y=x^{\pi}$ for some permutation $\pi$ of $[m]$.
Define another $\Gamma$-instance $\calI'$ with variables $\calV$ and the cost function
\begin{equation}
f_{\calI'}({\bf g})=f_\calI({\bf g}) \; + \hspace{-5pt}\sum_{(i,j,x,y)\in T}\hspace{-5pt}=_D({\bf g}(i,x),{\bf g}(j,y))
+ \hspace{-5pt}\sum_{(f,x)\in\Gamma^+}\hspace{-5pt} (\dom f)^m({\bf g}(x))
\quad \forall{\bf g}\in\calO^{(m\rightarrow m)}
\label{eq:instance:second}
\end{equation}
where $=_D$ is the equality relation on $D$.
The instance $\calI'$ is a $\Gamma$-instance because of condition (b) in the definition of a block-finite language.
Note that the second term in~\eqref{eq:instance:second} for ${\bf g}=(g_1,\ldots,g_m)$ equals $0$
if $g_{\pi(j)}(x)=g_j(x^\pi)$ for all $j\in[m]$, $x\in D^m$ and permutation $\pi$ of $[m]$. Otherwise the second term equals $\infty$.
In other words, the second term is zero if and only if  mapping ${\bf g}$ satisfies condition~\eqref{eq:Omega:def:a}, i.e.\ if and only if  ${\bf g}\in\Omega$.
Similarly, the third term in~\eqref{eq:instance:second} is zero if ${\bf g}\in\Omega'$,
and $\infty$ if ${\bf g}\in\Omega-\Omega'$.
We obtain that
\begin{subequations}
\label{eq:FALIHGALSF}
\begin{eqnarray}
f_{\calI'}({\mathds 1})&\le &f_{\calI'}({\bf g})\;<\;\infty\qquad \forall {\bf g}\in\mathbb G \\
f_{\calI'}({\mathds 1})&<&f_{\calI'}({\bf g})\;<\;\infty\qquad \forall {\bf g}\in\Omega'-\mathbb G \\
f_{\calI'}({\bf g})&=&\infty\hspace{46pt}\qquad \forall {\bf g}\in\calO^{(m\rightarrow m)}-\Omega'
\end{eqnarray}
\end{subequations}

%Finally, define a function $f\in\langle\Gamma\rangle$ with $m$ variables as follows:
%$$ f(x)=\min_{{\bf g}\in\calO^{(m\rightarrow m)}:{\bf g}(\hat x)=x}f_{\calI'}({\bf g})
%\qquad\forall x\in D^m
%$$
%We will show next that this function has the desired property.

These equations imply that ${\mathds 1}\in\arg\min f_{\calI'}\subseteq {\mathbb G}$.
We will show next that $\arg\min f_{\calI'}= {\mathbb G}$.

%First, observe that ${\mathds 1}\in\arg\min f_{\calI'}$ with $f_{\calI'}({\mathds 1})<\infty$, and therefore $\hat x\in\arg\min f$
%with $f(\hat x)<\infty$.
%From \eqref{eq:FALIHGALSF} we also get that if $x\in\arg\min f$ then there exists ${\bf g}\in\mathbb G$ such that ${\bf g}(\hat x)=x$.
%This implies that $\arg\min f\subseteq \mathbb G(\hat x)$.
%
%Consider $x\in\mathbb G(\hat x)$. Then $x={\bf h}(\hat x)$ for some ${\bf h}\in\mathbb G$.

%Let $\pi$ be the cyclic permutation of $[m]$ with $\pi(1)=2$.
%Then $\pi^i$ is the cyclic permutation of $[m]$ with $\pi(1)=1+i$ for $0\le i<m$.
For an index $k\in\mathbb Z$ let $\bar k\in[m]$ be the unique index with $\bar k-k=0~(\operatorname{mod}~m)$.
Let $\pi_k$ be the cyclic permutation of $[m]$ with $\pi_k(1)=\bar k$.
In particular, $\pi_1$ is the identity permutation.
Also, for $k \in \mathbb Z$ let $e^k\in\calO^{(m)}$ be the projection to the $\bar k$-th coordinate.
For a mapping ${\bf g}=(g_1,\ldots,g_m)\in\calO^{(m\rightarrow m)}$ and a tuple $z\in D^m$ we will denote ${\bf g}(k,z)={\bf g}(\bar k,z)\in D$.
From the definition, for any permutation $\pi$ of $[m]$ and any $(i,z)\in\calV$ we have ${\bf g}^\pi(i,z)=(g_{\pi(1)},\ldots,g_{\pi(m)})(i,z)=g_{\pi(i)}(z)={\bf g}(\pi(i),z)$. In particular, ${\bf g}^{\pi_j}(i,z)={\bf g}(i+j-1,z)$.

%For $i\in\mathbb Z$ let $\pi_i$ be the cyclic permutation of $[m]$ with $\pi_i(1)=i~(\operatorname{mod}~m)$.
%In particular, $\pi_1$ is the identity permutation.
From~\eqref{eq:instance:first':c} we have $f_\calI({\mathds 1}^{\pi_1})=\ldots=f_\calI({\mathds 1}^{\pi_m})=f_\calI({\mathds 1})$.
Recall that $f_{\calI}\in\langle\Gamma\rangle$ admits a generalized fractional polymorphism $\omega$ with $\supp(\omega)=\mathbb G$.
Applying this polymorphism gives
%Consider the following $m$ labelings for $\calI$:
% ${\mathds 1}^{\pi_1},\ldots,{\mathds 1}^{\pi_m}\in\arg\min f$.
%
\begin{equation}
\sum_{{\bf g}\in\supp(\omega)}\omega({\bf g})f^m_{\calI}({\bf g}({\mathds 1}^{\pi_1},\ldots,{\mathds 1}^{\pi_m}))\le f_\calI^m({\mathds 1}^{\pi_1},\ldots,{\mathds 1}^{\pi_m})=f_\calI({\mathds 1})
\label{eq:FASLASKFASFA}
\end{equation}
Here we view ${\mathds 1}^{\pi_1},\ldots,{\mathds 1}^{\pi_m}$ (and later ${\bf g}^{\pi_1},\ldots,{\bf g}^{\pi_m}$) as labelings for the instance $\calI$,
while ${\bf g}$ is a mapping in $\calO^{(m\rightarrow m)}$ acting on the first $m$ labelings coordinate-wise.
We claim that ${\bf g}({\mathds 1}^{\pi_1},\ldots,{\mathds 1}^{\pi_m})=({\bf g}^{\pi_1},\ldots,{\bf g}^{\pi_m})$
for each ${\bf g}=(g_1,\ldots,g_m)\in \supp(\omega)$. Indeed, we need to show that $g_j({\mathds 1}^{\pi_1},\ldots,{\mathds 1}^{\pi_m})={\bf g}^{\pi_j}$
for each $j\in[m]$.
Let us prove this for coordinate $(i,z)\in \calV$. We can write
\begin{eqnarray*}
g_j({\mathds 1}^{\pi_1}(i,z),\ldots,{\mathds 1}^{\pi_m}(i,z))
&=&g_j(e^i(z),\ldots,e^{i+m-1}(z)) \\
&=& g_j(z^{\pi_i})
\;=\; g_{\pi_i(j)}(z)
\;=\; {\bf g}(i+j-1,z)
\;=\; {\bf g}^{\pi_j}(i,z)
\end{eqnarray*}
which proves the claim. We can now rewrite~\eqref{eq:FASLASKFASFA} as follows:
\begin{equation}
\sum_{{\bf g}\in\supp(\omega)}\omega({\bf g})f^m_{\calI}({\bf g}^{\pi_1},\ldots,{\bf g}^{\pi_m})\le f_\calI({\mathds 1})
\end{equation}
Using \eqref{eq:instance:first':c} and the fact that $f_\calI({\bf g})=f_{\calI'}({\bf g})$ for each ${\bf g}\in\supp(\omega)={\mathbb G}$,
we obtain
\begin{equation}
\sum_{{\bf g}\in\supp(\omega)}\omega({\bf g})f_{\calI'}({\bf g})\le f_{\calI'}({\mathds 1})
\end{equation}
Since ${\mathds 1}\in\arg\min f_{\calI'}$, we conclude that
${\mathbb G}=\supp(\omega)\subseteq\arg\min f_{\calI'}$. Therefore, $\arg\min f_{\calI'}=\mathbb G$.

We can finally prove Theorem~\ref{th:specialInstance}.
We define function $f\in\closure\Gamma$ with $m$ variables as follows:
$$
f(x)=\min_{{\bf g}\in\calO^{(m\rightarrow m)}:{\bf g}(\hat x)=x}f_{\calI'}({\bf g})
\qquad\forall x\in D^m
$$
Consider tuple $x\in D^m$. We have $x\in\arg\min f$ if and only if  there exists ${\bf g}\in\arg\min f_{\calI'}=\mathbb G$ with ${\bf g}(\hat x)=x$.
The latter condition holds if and only if  $x\in\mathbb G(\hat x)$.
\end{proof}

%%%%%%%%%%%%%%%%%%%%%%%%%%%%%%%%%%%%%%%%%%%%%%%%%%%%%%%%%%%%%%%%%%%%%%%%%%%%%%%%%%%%%%%%%%%%%%%%%%%
%%%%%%%%%%%%%%%%%%%%%%%%%%%%%%%%%%%%%%%%%%%%%%%%%%%%%%%%%%%%%%%%%%%%%%%%%%%%%%%%%%%%%%%%%%%%%%%%%%%
%%%%%%%%%%%%%%%%%%%%%%%%%%%%%%%%%%%%%%%%%%%%%%%%%%%%%%%%%%%%%%%%%%%%%%%%%%%%%%%%%%%%%%%%%%%%%%%%%%%
%%%%%%%%%%%%%%%%%%%%%%%%%%%%%%%%%%%%%%%%%%%%%%%%%%%%%%%%%%%%%%%%%%%%%%%%%%%%%%%%%%%%%%%%%%%%%%%%%%%
%%%%%%%%%%%%%%%%%%%%%%%%%%%%%%%%%%%%%%%%%%%%%%%%%%%%%%%%%%%%%%%%%%%%%%%%%%%%%%%%%%%%%%%%%%%%%%%%%%%
%%%%%%%%%%%%%%%%%%%%%%%%%%%%%%%%%%%%%%%%%%%%%%%%%%%%%%%%%%%%%%%%%%%%%%%%%%%%%%%%%%%%%%%%%%%%%%%%%%%

\section{Proof of Theorem~\ref{th:symm}} \label{sec:thsym}

We will prove the following result.
\begin{theorem}
Assume that one of the following holds:
\begin{itemize}
\item[(a)] $m=2$ and $\Gamma$ admits a cyclic fractional polymorphism of arity at least 2.
\item[(b)] $m\ge 3$ and $\Gamma$ admits a symmetric fractional polymorphism of arity $m-1$.
\end{itemize}
Let $f\in\closure\Gamma$ be a function of arity $m$ with $\arg\min f=\mathbb G(\hat x)$, where $\hat x\in Range_1(\mathbb G^\ast)$.
Then for every distinct pair of indices $i,j\in[m]$ there exists $x\in\arg\min f$ with $x_i=x_j$.
%Fix $\hat x\in Range(\mathbb G^\ast)$ and distinct $i,j\in[m]$. Then $\hat x_i=\hat x_j$.
\label{th:rangeMain}
\end{theorem}
We claim that this will imply Theorem~\ref{th:symm}.
Indeed, we can use the following observation.
\begin{proposition}
Suppose that $\hat x\in Range_1(\mathbb G^\ast)$, and there exists $x\in\mathbb G(\hat x)$ with $x_i=x_j$ for some $i,j\in[m]$.
Then $\hat x_i=\hat x_j$.
\label{prop:GHAOGIAHSIFA}
\end{proposition}
\begin{proof}
By Proposition~\ref{prop:Ghatx}(b),
 there exists ${\bf g}\in\mathbb G$ such that ${\bf g}(x)=\hat x$.
%from the construction of graph $(\mathbb G,E)$ we know that have the following property for mappings ${\bf g}',{\bf g}''\in\mathbb G$:
Let $\pi$ be the permutation of $[m]$ that swaps $i$ and $j$. By the choice of $x$, we have $x^\pi=x$.
We can write $\hat x_j=g_j(x)=g_{\pi(i)}(x)=g_i(x^\pi)=g_i(x)=\hat x_i$. This proves the claim.
\end{proof}

\begin{corollary}
If the precondition of Theorem~\ref{th:rangeMain} holds,
then $\Gamma$ admits a symmetric fractional polymorphism of arity $m$.
\end{corollary}
\begin{proof}
Using Theorem~\ref{th:specialInstance}, Theorem~\ref{th:rangeMain} and Proposition~\ref{prop:GHAOGIAHSIFA},
we conclude that for any $\hat x\in Range_1(\mathbb G^\ast)$ we have $\hat x_1=\ldots=\hat x_m$.
Indeed, by Theorem~\ref{th:specialInstance} there exists a function $f\in\closure\Gamma$ with $\mathbb G(\hat x)=\arg\min f$.
Theorem~\ref{th:rangeMain} implies that the precondition of Proposition~\ref{prop:GHAOGIAHSIFA} holds for any distinct pair of indices $i,j\in[m]$, and therefore $\hat x_i=\hat x_j$.

By Theorem~\ref{th:G:rho}, there exists a generalized fractional polymorphism $\rho$ of $\Gamma$ of arity $m\rightarrow m$
with $\supp(\rho)=\mathbb G^\ast$. Vector $\sum_{{\bf g}=(g_1,\ldots,g_m)\in\mathbb G^\ast}\rho({\bf g})\frac{1}{m}[\chi_{g_1}+\ldots+\chi_{g_m}]$
is then an $m$-ary fractional polymorphism of $\Gamma$; all operations in its support are symmetric because $\mathbb G^\ast\subseteq \Omega$ and
$\hat x_1=\ldots=\hat x_m$ for any $\hat x\in Range_1(\mathbb G^\ast)$.
\end{proof}

It remains to prove Theorem~\ref{th:rangeMain}.
A proof of parts (a) and (b) of Theorem~\ref{th:rangeMain} is given in Sections~\ref{sec:sym2} and~\ref{sec:symm}, respectively.
In both parts we will need the following result; it exploits the fact that $\Gamma$ is block-finite.
%Since $\Gamma$ is block-finite, it will imply that $(a,\ldots,a)\in \dom f$ for any $a\in\{\hat x_1,\ldots,\hat x_m\}$
%for the instance $f$ used in Theorem~\ref{th:rangeMain}.
\begin{lemma}
Suppose that $\hat x\in Range_1(\mathbb G^\ast)$, $x\in \mathbb G(\hat x)$ and $f$ is an $m$-ary function in $\closure\Gamma$
with $\arg\min f=\mathbb G(\hat x)$. Then $(a,\ldots,a)\in\dom f$ for any $a\in\{x_1,\ldots,x_m\}$.
%If $\hat x\in Range_1(\mathbb G^\ast)$ then there exists $v\in V$ such that $\hat x_1,\ldots,\hat x_m\in D_v$.
\label{lemma:RangeInTheSameBlock}
\end{lemma}
\begin{proof}
We say that a tuple $z\in D^m$ is {\em proper} if $z_1,\ldots,z_m\in D_v$ for some $v\in V$.
We will show that  $x$ is proper;
the lemma will then follow from
condition~(a) from the definition of a block-finite language and the fact that $x\in\dom f$.

Fix an arbitrary element $a\in D$, and define mapping ${\bf g}\in\calO^{(m\rightarrow m)}$ as follows:
$$
{\bf g}(z)=\begin{cases}
z & \mbox{if $z$ is proper} \\
(a,\ldots,a) & \mbox{otherwise}
\end{cases}
$$
We claim that ${\bf g}\in\Omega$. Indeed, consider $z \in D^m$. If ${\bf g}(z) = z$, the condition (\ref{eq:Omega:def:a}) holds trivially. Otherwise, we can easily check that
$${\bf g}^\pi(z) = \left(a,\ldots,a\right)^\pi = ( a,\ldots,a) = {\bf g}(z^\pi)$$
and so the condition (\ref{eq:Omega:def:a}) holds either way.

Let us now show that the vector $\rho=\chi_{\bf g}$ is a generalized fractional polymorphism of $\Gamma$ of arity $m\rightarrow m$.
Checking inequality~\eqref{eq:genfpol} for binary equality relation $f=(=_D)$ is straighforward. Consider function $f\in\Gamma-\{=_D\}$.
Since $\Gamma$ is block-finite, we have $\dom f\subseteq D_{v_1}\times\ldots\times D_{v_n}$  for some $v_1,\ldots,v_n\in V$.
This implies that for any $x\in[\dom f]^m$ we have ${\bf g}(x)=x$ (this can be checked coordinate-wise).
Therefore, we have an equality in~\eqref{eq:genfpol}.

By the results above we obtain that ${\bf g}\in\mathbb G$.
We are now ready to prove that $x$ is proper.
Suppose that this is not true, then ${\bf g}(x)=(a,\ldots,a)$.
We have $\hat x\in Range_1(\mathbb G^\ast)$ and $x\in\mathbb G(\hat x)$, so by Proposition~\ref{prop:Ghatx}(a) we conclude that
 $x\in Range_1(\mathbb G^\ast)$.
We also have $(a,\ldots,a)\in\mathbb G(x)$, so Proposition~\ref{prop:GHAOGIAHSIFA}
gives that $x_1=\ldots=x_m$. This means that $x$ is proper, which contradicts the earlier assumption.
\end{proof}

%Since $\Gamma$ is block-finite, we conclude that $(a,\ldots,a)\in \dom f$ for any $a\in\{\hat x_1,\ldots,\hat x_m\}$
%where $f$ is the instance in Theorem~\ref{th:rangeMain}.

%In both parts we assume that $f\in\closure\Gamma$ is the function of arity $m$ with $\arg\min f=\mathbb G(\hat x)$.
%Note, condition $\hat x\in Range_1(\mathbb G^\ast)$ implies that elements $\hat x_1,\ldots,\hat x_m$
%belong to the same set $D_v$ for some $v\in V$. (This can be concluded
%from condition (d) of Definition~\ref{def:block-finite}.)
%Since $\Gamma$ is block-finite,
%we obtain that $(a,\ldots,a)\in \dom f$ for any $a\in\{\hat x_1,\ldots,\hat x_m\}$.
%%%%%%%%%%%%%%%%%%%%%%%%%%%%%%%%%%%%%%%%%%%%%%%%%%%%%%%%%%%%%%%%%%%%%%%%%%%%%%%%%%%%%%%%%%%%%%%%%%%
%%%%%%%%%%%%%%%%%%%%%%%%%%%%%%%%%%%%%%%%%%%%%%%%%%%%%%%%%%%%%%%%%%%%%%%%%%%%%%%%%%%%%%%%%%%%%%%%%%%
%%%%%%%%%%%%%%%%%%%%%%%%%%%%%%%%%%%%%%%%%%%%%%%%%%%%%%%%%%%%%%%%%%%%%%%%%%%%%%%%%%%%%%%%%%%%%%%%%%%
%%%%%%%%%%%%%%%%%%%%%%%%%%%%%%%%%%%%%%%%%%%%%%%%%%%%%%%%%%%%%%%%%%%%%%%%%%%%%%%%%%%%%%%%%%%%%%%%%%%
%%%%%%%%%%%%%%%%%%%%%%%%%%%%%%%%%%%%%%%%%%%%%%%%%%%%%%%%%%%%%%%%%%%%%%%%%%%%%%%%%%%%%%%%%%%%%%%%%%%
%%%%%%%%%%%%%%%%%%%%%%%%%%%%%%%%%%%%%%%%%%%%%%%%%%%%%%%%%%%%%%%%%%%%%%%%%%%%%%%%%%%%%%%%%%%%%%%%%%%
\subsection{Case $m=2$: proof of Theorem~\ref{th:rangeMain}(a)}\label{sec:sym2}

%For a operation $g\in\calO^{(2)}$ we define operation $\bar g\in\calO^{(2)}$ via $\bar g(a,b)=g(b,a)$.
We start with the following observation.
\begin{proposition}
%(a) Every ${\bf g}\in\mathbb G$ has the form  ${\bf g}=(g,\bar g)$ for some $g\in\calO^{(2)}$.
%(b)
If $(a,b)\in\mathbb G(\hat x)$ then $(b,a)\in\mathbb G(\hat x)$.
\label{prop:m2-property}
\end{proposition}
\begin{proof}
Consider mapping $\bar{\mathds 1}=(e_2^2,e_2^1)$, where $e_2^k\in\calO^{(2)}$ is the projection to the the $k$-th variable.
It can be checked that $\bar{\mathds 1}\in\Omega$, and $\chi_{\bar{\mathds 1}}$ is a generalized fractional polymorphism of $\Gamma$
of arity $2\rightarrow 2$. Therefore, $\bar{\mathds 1}\in\mathbb G$.

We have $(a,b)={\bf g}(\hat x)$ for some ${\bf g}\in\mathbb G$. We also have $(b,a)=(\bar{\mathds 1}\circ{\bf g})(\hat x)$ and $\bar{\mathds 1}\circ{\bf g}\in\mathbb G$,
and therefore $(b,a)\in\mathbb G(\hat x)$.
\end{proof}

Denote $A=\{x_1\:|\:x\in \mathbb G(\hat x)\}\subseteq D$,
and let $a$ be an element in $A$ that minimizes $f(a,a)$.
Note that $(a,a)\in\dom f$ by Lemma~\ref{lemma:RangeInTheSameBlock}.
Condition $\arg\min f=\mathbb G(\hat x)$ and Proposition~\ref{prop:m2-property} imply that $(a,b),(b,a)\in\arg\min f$
for some $b\in A$.
By assumption, $\Gamma$ admits a cyclic fractional polymorphism $\nu$ of some arity $r\ge 2$.
Let us apply it to tuples $(a,b),(b,a),(a,a),\ldots,(a,a)$, where $(a,a)$ is repeated $r-2$ times:
\begin{equation}
\sum_{h\in\supp(\nu)} \nu(h) f(h(a,b,a,\ldots,a),h(b,a,a,\ldots,a))\le \frac{2}{r} f(a,b)+\frac{r-2}{r} f(a,a)
\label{eq:GHAJDHGAKSFA}
\end{equation}
We have $h(a,b,a,\ldots,a)=h(b,a,a,\ldots,a)$ since $\nu$ is cyclic; denote this element as $a_h$.
We claim that $a_h\in A$ for any $h\in\supp(\nu)$. Indeed, consider a unary function $u_A(x_1)=\min_{x_2}  f(x_1,x_2)$.
It can be checked that $\arg\min u_A=A$. Then
the presence of $u_A$ in $\closure\Gamma$ implies that after applying $\nu$ to $(a, b, a, \dots, a)$ one gets
$$\sum_{h\in\supp(\nu)} \nu(h)u_A(a_h) \leq \frac{r-1}{r} u_A(a) + \frac{1}{r}u_A(b) = \min u_A$$
and thus indeed $a_h\in \arg\min u_A = A$ for any $h\in\supp(\nu)$.

By the choice of $a$ we have $ f(a,a)\le  f(a_h,a_h)$ for any $h\in\supp(\nu)$.
From~\eqref{eq:GHAJDHGAKSFA} we thus get
\begin{equation}
 f(a,a)\le \frac{2}{r} f(a,b)+\frac{r-2}{r} f(a,a)
\end{equation}
and so $ f(a,a)\le  f(a,b)$, implying $(a,a)\in\arg\min  f$.
%%%%%%%%%%%%%%%%%%%%%%%%%%%%%%%%%%%%%%%%%%%%%%%%%%%%%%%%%%%%%%%%%%%%%%%%%%%%%%%%%%%%%%%%%%%%%%%%%%%
%%%%%%%%%%%%%%%%%%%%%%%%%%%%%%%%%%%%%%%%%%%%%%%%%%%%%%%%%%%%%%%%%%%%%%%%%%%%%%%%%%%%%%%%%%%%%%%%%%%
%%%%%%%%%%%%%%%%%%%%%%%%%%%%%%%%%%%%%%%%%%%%%%%%%%%%%%%%%%%%%%%%%%%%%%%%%%%%%%%%%%%%%%%%%%%%%%%%%%%
%%%%%%%%%%%%%%%%%%%%%%%%%%%%%%%%%%%%%%%%%%%%%%%%%%%%%%%%%%%%%%%%%%%%%%%%%%%%%%%%%%%%%%%%%%%%%%%%%%%
%%%%%%%%%%%%%%%%%%%%%%%%%%%%%%%%%%%%%%%%%%%%%%%%%%%%%%%%%%%%%%%%%%%%%%%%%%%%%%%%%%%%%%%%%%%%%%%%%%%
%%%%%%%%%%%%%%%%%%%%%%%%%%%%%%%%%%%%%%%%%%%%%%%%%%%%%%%%%%%%%%%%%%%%%%%%%%%%%%%%%%%%%%%%%%%%%%%%%%%

\subsection{Case $m\ge 3$: proof of Theorem~\ref{th:rangeMain}(b)}\label{sec:symm}

We define binary function $\bar f\in\closure\Gamma$ as follows: $\bar f(a,b)=\min_{x\in D^m:x_i=a,x_j=b} f(x)$.

If $z=(z_1,\ldots,z_m)$ is some sequence of size $m$ and $k$ is an index in $[m]$
then we will use $z_{-k}$ to denote the subsequence of $z$ of size $m-1$
obtained by deleting the $k$-th element.

Let $\tilde\omega$ be a symmetric fractional polymorphism of $\Gamma$ of arity $m-1$.
Following the construction in~\cite{kolmogorov15:power}, we define graph $(\tilde{\mathbb G},\tilde E)$ as described in Section~\ref{sec:G},
starting with the distribution $\tilde\omega$
where for $s\in\supp(\tilde\omega)$ mapping ${\mathds 1}^s\in\calO^{(m\rightarrow m)}$ is defined as follows:
$$
{\mathds 1}^s(x)=(s(x_{-1}),\ldots,s(x_{-m}))\qquad\forall x\in D^m
$$
It can be checked that if ${\bf g}=(g_1,\ldots,g_m)\in\tilde{\mathbb G}$
and $s\in\supp(\tilde\omega)$ then
${\bf g}^s=(s\circ{\bf g}_{-1},\ldots,s\circ{\bf g}_{-m})$.
It can also be checked that condition~\eqref{eq:G:omega'} holds for any $f\in\Gamma$:
it corresponds to the fractional polymorphism $\tilde\omega$ applied to $m-1$ tuples $x_{-i}\in [\dom f]^{m-1}$.
\begin{proposition}
There holds $\tilde{\mathbb G}\subseteq\mathbb G$. \label{prop:subset}
\end{proposition}
\begin{proof}
We claim that ${\mathds 1}^s\in\Omega$ for any $s\in\supp(\tilde\omega)$. Indeed, for a permutation $\pi$ of $[m]$ and $x \in D^m$ we can write
%First, let us verify the property (\ref{eq:Omega:def:a}) for ${\bf g} = {\mathds 1}^{s_1\dots s_k} \in \tilde{\mathbb G}$. Fix a permutation $\pi$ of $[m]$ and $x \in D^m$ and note that
$$
{\mathds 1}^s(x^\pi)=(s(x^\pi_{-1}),\ldots,s(x^\pi_{-m})) =  (s(x_{-\pi(1)}),\ldots,s(x_{-\pi(m)}))  = ({\mathds 1}^s)^\pi(x),
$$
where the second equality uses that $s$ is symmetric.
Since each ${\bf g} \in \tilde{\mathbb G}$ has the form ${\bf g} = {\mathds 1}^{s_k}\circ\ldots\circ{\mathds 1}^{s_1}$ for some $s_1,\ldots,s_k\in\supp(\tilde\omega)$ and $\Omega$ is closed under composition by Proposition \ref{prop:Omegaclosed}, we get $\tilde{\mathbb G} \subseteq \Omega$.

Applying Theorem \ref{th:G:rho} (with $\tilde{\mathbb G}$ as both $\mathbb G$ and $\widehat{\mathbb G}$), we obtain a generalized fractional polymorphism $\tilde\rho$ with $\supp(\tilde\rho) = \tilde{\mathbb G} \subseteq \Omega$. By maximality of $\omega$ we get the desired $\tilde{\mathbb G} = \supp(\tilde\rho) \subseteq \supp(\omega) = \mathbb G$.
%
%Then
%$${\bf g}(x^\pi) = {\mathds 1}^{s_1\dots s_k}(x^\pi) = {\mathds 1}^{s_2\dots s_k}\left( ({\mathds 1}^{s_1})(x^\pi)\right) = {\mathds 1}^{s_2\dots s_k}\left( ({\mathds 1}^{s_1})^\pi(x) \right) = \dots = \left({\mathds 1}^{s_1\dots s_k}\right)^\pi(x) = {\bf g}^\pi(x).$$
%
%The rest of the proposition is established by using Theorem \ref{th:G:rho} with $\widehat{\tilde{\mathbb G}} = \tilde{\mathbb G}$.
\end{proof}

%and $\tilde{\mathbb G}^\ast(\hat x)=\{{\bf g}(\hat x)\in\tilde{\mathbb G}^\ast\}$.
For each ${\bf g}\in \tilde{\mathbb G}$ and $k\in[m]$ let us define labeling $x^{[{\bf g}k]}\in D^2$ as follows:
set $x={\bf g}(\hat x)$, and then
\begin{itemize}
\item If $k=i$, set $x^{[{\bf g}k]}=(x_i,x_j)$. We have $x^{[{\bf g}i]}\in\arg\min \bar f$ since $x\in\mathbb G(\hat x)=\arg\min f$.
\item If $k=j$, set $x^{[{\bf g}k]}=(x_j,x_i)$.
\item If $k\ne i$ and $k\ne j$, set $x^{[{\bf g}k]}=(x_k,x_k)$. We have $x^{[{\bf g}k]}\in \dom\bar f$ by Lemma~\ref{lemma:RangeInTheSameBlock}.
\end{itemize}
%Note that $x\in\mathbb G(\hat x)$, and therefore by construction $x^{[{\bf g}i]}\in\arg\min f'$.
\begin{proposition}
%For any   ${\bf g}\in \tilde{\mathbb G}$ there holds ${\bf g}(x)=(x^{[{\bf g}1]},\ldots,x^{[{\bf g}m]})$.
Suppose that ${\bf g}\in \tilde{\mathbb G}$ and ${\bf g}^s={\bf h}$ where $s\in\supp(\tilde\omega)$ (so that ${\bf h}\in\tilde{\mathbb G}$).
Then
$$
{\mathds 1}^s(x^{[{\bf g}1]},\ldots,x^{[{\bf g}m]})=(x^{[{\bf h}1]},\ldots,x^{[{\bf h}m]}).
$$
\label{prop:GAIHGASFHAISG}
\end{proposition}
\begin{proof}
Denote $x={\bf g}(\hat x)$ and $y={\bf h}(\hat x)$. We have $y={\mathds 1}^s(x)$, or $y_k=s(x_{-k})$ for any $k\in[m]$.
Also,
$$
x^{[{\bf g}k]}=\begin{cases}
(x_i,x_j) & \mbox{if }k=i \\
(x_j,x_i) & \mbox{if }k=j \\
(x_k,x_k) & \mbox{if }k\ne i\mbox{ and }k\ne j
\end{cases}
\qquad
x^{[{\bf h}k]}=\begin{cases}
(y_i,y_j) & \mbox{if }k=i \\
(y_j,y_i) & \mbox{if }k=j \\
(y_k,y_k) & \mbox{if }k\ne i\mbox{ and }k\ne j
\end{cases}
$$
It can be checked coordinate-wise (using that $s$ is symmetric) that  $$x^{[{\bf h}k]}=s((x^{[{\bf g}1]},\ldots,x^{[{\bf g}m]})_{-k})$$ for any $k\in[m]$. This gives the claim.
\end{proof}
Denote
$\tilde{\mathbb G}^\ast=\bigcup_{\mathbb H\in {\sf Sinks}({\tilde{\mathbb G},\tilde E})} \mathbb H\subseteq \tilde{\mathbb G}$.
%Let us a fix some component ${\mathbb H}\in {\sf Sinks}({\tilde{\mathbb G},\tilde E})$ and a mapping $\tilde{\bf g}\in\mathbb H$.
Let us fix an arbitrary $\tilde{\bf g}\in \tilde{\mathbb G}^\ast$, and define $\tilde x=(x^{[\tilde{\bf g} 1]},\ldots,x^{[\tilde{\bf g} m]})\in [D^2]^m$.
%\begin{proposition}
%For any   ${\bf g}\in \tilde{\mathbb G}$ there holds ${\bf g}(\tilde x)=(x^{[{\bf g}1]},\ldots,x^{[{\bf g}m]})$.
%Suppose that ${\bf g}\in \tilde{\mathbb G}^\ast$ and $s\in\supp(\tilde\omega)$, and so ${\bf h}={\bf g}^s\in \tilde{\mathbb G}^\ast$.
%Then $$
%(x^{[{\bf h}1]},\ldots,x^{[{\bf h}m]})=({\mathds 1}^s)(x^{[{\bf g}1]},\ldots,x^{[{\bf g}m]})
%$$
%\end{proposition}
%\begin{proof}
%TODO
%\end{proof}

\begin{proposition}
For any   ${\bf g}\in \tilde{\mathbb G}$ there holds
 ${\bf g}\circ\tilde{\bf g}\in\tilde{\mathbb G}$. Furthermore,
${\bf g}(\tilde x)=(x^{[({\bf g}\circ\tilde{\bf g})1]},\ldots,x^{[({\bf g}\circ\tilde{\bf g})m]})$.
%Suppose that ${\bf g}\in \tilde{\mathbb G}^\ast$ and $s\in\supp(\tilde\omega)$, and so ${\bf h}={\bf g}^s\in \tilde{\mathbb G}^\ast$.
%Then $$
%(x^{[{\bf h}1]},\ldots,x^{[{\bf h}m]})=({\mathds 1}^s)(x^{[{\bf g}1]},\ldots,x^{[{\bf g}m]})
%$$
\label{prop:almostDone}
\end{proposition}
\begin{proof}
The first claim is by Proposition~\ref{prop:Range}(a); let us show the second one.
Let $d({\mathds 1},{\bf g})$ be the shortest distance from ${\mathds 1}$ to ${\bf g}$ in the graph $(\tilde{\mathbb G},\tilde E)$.
(By the definition of this graph, we have $0\le d({\mathds 1},{\bf g}) <\infty$ for any ${\bf g}\in\tilde{\mathbb G}$, and ${\mathds 1}\in\tilde{\mathbb G}$.)
We will use induction on $d({\mathds 1},{\bf g})$. The base case $d({\mathds 1},{\bf g})=0$ (i.e.\ ${\bf g}={\mathds 1}$) holds by construction.
Suppose that the claim holds for all mappings ${\bf g}\in\tilde{\mathbb G}$ with $d({\mathds 1},{\bf g})=k\ge 0$,
and consider mapping ${\bf h}\in\tilde{\mathbb G}$ with $d({\mathds 1},{\bf h})=k+1$.
There must exist mapping ${\bf g}\in\tilde{\mathbb G}$ and operation $s\in\supp(\tilde\omega)$
such that $d({\mathds 1},{\bf g})=k$ and ${\bf g}^s={\bf h}$.
Observe that $({\bf g}\circ\tilde{\bf g})^s={\mathds 1}^s\circ {\bf g}\circ\tilde{\bf g}={\bf g}^s\circ\tilde{\bf g}={\bf h}\circ\tilde{\bf g}$.
We can thus write
$$
{\bf h}(\tilde x)
=({\mathds 1}^{s}\circ{\bf g})(\tilde x)
= {\mathds 1}^{s}({\bf g}(\tilde x))
\stackrel{\mbox{\tiny(1)}}={\mathds 1}^{s}(x^{[({\bf g}\circ\tilde{\bf g})1]},\ldots,x^{[({\bf g}\circ\tilde{\bf g})m]})
\stackrel{\mbox{\tiny(2)}}=               (x^{[({\bf h}\circ\tilde{\bf g})1]},\ldots,x^{[({\bf h}\circ\tilde{\bf g})m]})
$$
where (1) holds by the induction hypothesis and (2) is by Proposition~\ref{prop:GAIHGASFHAISG}.
\end{proof}

\begin{proposition}
There holds $\tilde x\in Range_2(\tilde{\mathbb G}^\ast)\cap[\dom \bar f]^m$. %For each $k\in[m]$ there holds $x^{[\tilde{\bf g} k]}\in\dom \bar f$.
\end{proposition}
\begin{proof}
By Proposition~\ref{prop:Range}(b) there exists ${\bf g}\in\tilde{\mathbb G}^\ast$ with ${\bf g}\circ\tilde{\bf g}=\tilde{\bf g}$.
Using  Proposition~\ref{prop:almostDone}, we can  write ${\bf g}(\tilde x)=(x^{[({\bf g}\circ\tilde{\bf g})1]},\ldots,x^{[({\bf g}\circ\tilde{\bf g})m]})=(x^{[\tilde{\bf g}1]},\ldots,x^{[\tilde{\bf g}m]})=\tilde x$. This shows that $\tilde x\in Range_2(\tilde{\mathbb G}^\ast)$.

Now let us show $x^{[\tilde{\bf g} k]}\in\dom \bar f$ for each $k\in[m]$.
It suffices to prove it
 for $k=j$
(for other indices $k$ the claim holds by construction).
We have $\tilde{\bf g}\in\mathbb H$ for some strongly connected component $\mathbb H\in{\sf Sinks}(\tilde{\mathbb G},\tilde E)$.
There is a path from $\tilde{\bf g}$ to $\tilde{\bf g}$ in $(\mathbb H,E[\mathbb H])$, therefore
there exists mapping ${\bf h}\in \mathbb H\subseteq\tilde{\mathbb G}^\ast$ and $s\in\supp(\tilde\omega)$ with ${\bf h}^s=\tilde{\bf g}$.
Define $x=(x^{[{\bf h}1]},\ldots,x^{[{\bf h}m]})$, then by Proposition \ref{prop:GAIHGASFHAISG} we have ${\mathds 1}^s(x)=\tilde x$.
In particular,  $x^{[\tilde{\bf g} j]}=s(x_{-j})$.
Also, we have $x_{-j}\in[\dom \bar f]^{m-1}$ by construction.
Since $\Gamma$ admits $\tilde\omega$ and $s\in\supp(\tilde\omega)$, we conclude that $x^{[\tilde{\bf g} j]}\in \dom \bar f$.
\end{proof}

Pick $k\in[m]-\{i,j\}$.
By Theorem~\ref{th:G}(b) we obtain that there exists a probability distribution $\lambda$ over $\tilde{\mathbb G}^\ast$
such that $\bar f^\lambda_i(\tilde x)=\bar f^\lambda_k(\tilde x)$. Using Proposition~\ref{prop:almostDone},
we can rewrite this condition as
$$
\sum_{{\bf g}\in\tilde{\mathbb G}^\ast}\lambda_{\bf g}\bar f(x^{[({\bf g}\circ\tilde{\bf g})i]})
=\sum_{{\bf g}\in\tilde{\mathbb G}^\ast}\lambda_{\bf g}\bar f(x^{[({\bf g}\circ\tilde{\bf g})k]})
$$
Every tuple $x^{[({\bf g}\circ\tilde{\bf g})i]}$ on the LHS belongs to $\arg\min \bar f$.
Therefore, every tuple $x^{[({\bf g}\circ\tilde{\bf g})k]}$ on the RHS corresponding to mapping ${\bf g}\in\tilde{\mathbb G}^\ast$ with $\lambda_{\bf g}>0$
also belongs to $\arg\min \bar f$.

We proved that there exists $x\in\arg\min f$ with $x_i=x_j$.

\appendix
\section*{APPENDIX: Proofs for Section \ref{sec:G}}\label{sec:proof:G}
\setcounter{section}{1}

In this section we prove the properties of graph $(\mathbb G,E)$ stated in Section~\ref{sec:G}.

\subsection{Proof of Proposition~\ref{prop:Range}}\label{sec:proof:prop:Range}

\paragraph{Part (a)} We have ${\bf g}={\mathds 1}^{s_1\ldots s_{k}}$ and ${\bf h}={\mathds 1}^{s_{{k+1}}\ldots s_{\ell}}$
for some $s_1,\ldots,s_\ell\in\supp(\omega)$ and $0\le k\le \ell$.
Therefore,
${\bf h}\circ{\bf g}=[{\mathds 1}^{s_\ell}\circ\ldots\circ {\mathds 1}^{s_{k+1}}]\circ[{\mathds 1}^{s_k}\circ\ldots\circ{\mathds 1}^{s_1}]
={\mathds 1}^{s_1\ldots s_\ell}\in\mathbb G$. Also, ${\bf h}\circ{\bf g}={\bf g}^{s_{k+1}\ldots s_\ell}$,
and so there is path from ${\bf g}$ to ${\bf h}\circ{\bf g}$ in $(\mathbb G,E)$.
Since no edges leave the strongly connected component $\mathbb H$, we obtain that if ${\bf g}\in\mathbb H$ then ${\bf h}\circ{\bf g}\in\mathbb H$.

\paragraph{Part (b)}
%Note that ${\mathds 1}^{s_1\ldots s_k}\circ {\bf h}={\bf h}^{s_1\ldots s_k}$ for any $s_1,\ldots,s_k\in\supp(\omega)$
%and ${\bf h}\in\mathbb G$. Therefore, conditions ${\bf g}\in\mathbb G$, ${\bf h}\in\mathbb H'\in{\sf Sinks}({\mathbb G,E})$
%imply that ${\bf g}\circ{\bf h}\in\mathbb H'$ (since ${\bf g}$ can be written as ${\bf g}={\mathds 1}^{s_1\ldots s_k}$
%and there are no edges leaving ${\mathbb H}'$).
Pick $\hat{\bf g}\in\mathbb H$. Since $\mathbb H'$ is strongly connected,
there is a path from $\hat{\bf g}\circ{\bf g}'\in\mathbb H'$ to ${\bf g}'\in\mathbb H'$
in $(\mathbb G,E)$, i.e.\
${\bf g}'=[\hat{\bf g}\circ{\bf g}']^{s_1\ldots s_k}={\bf h}\circ \hat{\bf g}\circ{\bf g}'$
where ${\bf h}={\mathds 1}^{s_1\ldots s_k}$. It can be checked that mapping ${\bf g}={\bf h}\circ \hat{\bf g}$
has the desired properties.

\paragraph{Part (c)} By assumption, $x={\bf g}'(y)$ for some ${\bf g}'\in\mathbb H'\in{\sf Sinks}({\mathbb G,E})$ and $y\in[D^n]^m$.
By part (b) there exists ${\bf g}\in\mathbb H$ satisfying ${\bf g}\circ{\bf g}'={\bf g}'$.
We get that ${\bf g}(x)={\bf g}({\bf g}'(y))=({\bf g}\circ{\bf g}')(y)={\bf g}'(y)=x$.

\subsection{Proof of Proposition \ref{prop:Ghatx}}\label{sec:proof:prop:Ghatx}
By assumption, we have $\hat x={\bf g}^\ast(y)$ for some ${\bf g}^\ast\in\mathbb G^\ast$, $y\in D^m$
and $x={\bf h}(\hat x)$ for some ${\bf h}\in\mathbb G$.

\paragraph{Part (a)}
We have $x=({\bf h}\circ{\bf g}^\ast)(y)$ with ${\bf h}\circ{\bf g}^\ast\in\mathbb G^\ast$;
this establishes the claim.

\paragraph{Part (b)}
Let $\mathbb H\in{\sf Sinks}({\mathbb G},E)$ be the strongly connected component to which ${\bf g}^\ast$ belongs.
There exists a path in $(\mathbb H,E[\mathbb H])$ from ${\bf h}\circ{\bf g}^\ast\in\mathbb H$ to ${\bf g}^\ast\in\mathbb H$,
i.e.\ ${\bf g}^\ast={\mathds 1}^{s_1\ldots s_k}\circ {\bf h}\circ{\bf g}^\ast$ for some $s_1,\ldots,s_k\in\supp(\omega)=\mathbb G$.
Define ${\bf g}={\mathds 1}^{s_1\ldots s_k}\in\mathbb G$, then
${\bf g}^\ast = {\bf g}\circ{\bf h}\circ{\bf g}^\ast$. We have
$\hat x={\bf g}^\ast(y) = ({\bf g}\circ{\bf h}\circ{\bf g}^\ast)(y)= ({\bf g}\circ{\bf h})(\hat x) = {\bf g}(x)$,
as claimed.

\subsection{Proof of Theorem~\ref{th:G:rho}}\label{sec:proof:G:rho}
First, we make the following observation.
\begin{proposition}
Suppose vector $\rho$ is a fractional polymorphism of $\Gamma$ of arity $m\rightarrow m$
and ${\bf g}\in\supp(\rho)$.
Then the following vector is also a fractional polymorphism of $\Gamma$ of arity $m\rightarrow m$:
\begin{equation}
\rho[{\bf g}]\ = \ \rho+\frac{\rho({\bf g})}{2}\left[-\chi_{\bf g}+\sum_{s\in\omega}\omega(s)\chi_{{\bf g}^s}\right]
\label{eq:proof:rho:G:rho:rhog}
\end{equation}
\end{proposition}
\begin{proof}
Denote the vector in the square brackets as $\delta$.
Consider function $f\in\Gamma$ and labeling $x\in[\dom f]^m$. Since $\rho$ is a fractional polymorphism of $\Gamma$,
we have ${\bf g}(x)\in[\dom f]^m$. We can write
%\begin{eqnarray}
%\sum_{s\in\supp(\omega)} \omega(s)f^m({\bf g}^s(x))&\le& f^m({\bf g}(x))\qquad \forall x\in  [\dom f]^m
%\end{eqnarray}
%Denote $\rho'=\sum_{s\in\supp(\omega)}\omega(s)\chi_{{\bf s}^s}$, then
$$
\sum_{{\bf h}\in \supp(\rho[{\bf g}])}\delta({\bf h})f^m({\bf h}(x))
=-f^m({\bf g}(x))+\sum_{s\in\supp(\omega)}\omega(s)f^m({\bf g}^s(x))\le 0
$$
where the last inequality follows from condition~\eqref{eq:G:omega} applied to labelings ${\bf g}(x)$.
Thus, adding the extra term to $\rho$ in~\eqref{eq:proof:rho:G:rho:rhog} will not violate the
fractional polymorphism inequality for any $x\in[\dom f]^m$.
\end{proof}
Note that $\supp(\rho[{\bf g}])=\supp(\rho)\cup\{{\bf g}^s\:|\:s\in\supp(\omega)\}$ for ${\bf g}\in\supp(\rho)$.

We claim that $\Gamma$ admits a fractional polymorphism $\widehat \rho$ with $\supp(\widehat \rho)=\mathbb G$.
Indeed, we can start with vector $\rho=\chi_{\mathds 1}$ and then repeatedly modify it as $\rho\leftarrow \rho[{\bf g}]$
for mappings ${\bf g}\in\supp(\rho)$ that haven't appeared before; after $|\mathbb G|-1$ steps
we get a vector $\widehat \rho$ with the claimed property.

Let $\Omega$ be the set of fractional polymorphisms $\rho$ of $\Gamma$ with $\supp(\rho)\subseteq\mathbb G$
that satisfy $\rho({\bf g})\ge \widehat\rho({\bf g})$ for all ${\bf g}\in\widehat{\mathbb G}$.
Set $\Omega$ is non-empty since it contains $\widehat \rho$.
Let $\rho$ be a vector in $\Omega$ that maximizes $\rho(\widehat{\mathbb G})=\sum_{{\bf g}\in\widehat{\mathbb G}}\rho({\bf g})$.
(This maximum is attained since $\Omega$ is a compact subset of $\mathbb R^{|\mathbb G|}$). We claim
that $\supp(\rho)=\widehat{\mathbb G}$. Indeed, the inclusion $\widehat{\mathbb G}\subseteq \supp(\rho)$ is by construction.
Suppose there exists ${\bf g}\in\supp(\rho)-\widehat{\mathbb G}$. By the condition of Theorem~\ref{th:G:rho}
there exists a path ${\bf g}_0,\ldots,{\bf g}_k$ in $(\mathbb G,E)$ from ${\bf g}_0={\bf g}$
such that ${\bf g}_0,\ldots,{\bf g}_{k-1}\in\mathbb G-\widehat{\mathbb G}$ and
 ${\bf g}_k\in\widehat{\mathbb G}$.
It can be checked that vector $\rho'=\rho[{\bf g}_0]\ldots[{\bf g}_{k-1}]$
satisfies $\rho'\in\Omega$, $\rho'({\bf g})\ge \rho({\bf g})$ for ${\bf g}\in\widehat{\mathbb G}$,
and $\rho'({\bf g}_k)> \rho({\bf g}_k)$. This contradicts the choice of $\rho$.

\subsection{Proof of Theorem~\ref{th:G}(a)}\label{sec:proof:th:Ga}
Consider component $\mathbb H\in{\sf Sinks}({\mathbb G,E})$, and
denote $\mathbb H^\ast=\arg\min\{f^m({\bf g}(x))\:|\:{\bf g}\in\mathbb H\}$.
We claim that $\mathbb H^\ast=\mathbb H$. Indeed, consider
${\bf g}\in\mathbb H^\ast$. Applying inequality~\eqref{eq:G:omega}
to labelings ${\bf g}(x)\in[\dom f]^m$ gives
\begin{eqnarray}
\sum_{s\in\supp(\omega)} \omega(s)f^m({\bf g}^s(x))&\le& f^m({\bf g}(x))\qquad \forall x\in  [\dom f]^m
\end{eqnarray}
For each $s\in\supp(\omega)$ we have ${\bf g}^s\in\mathbb H$ and thus  $f^m({\bf g}^s(x)) \ge f^m({\bf g}(x))$.
This means that $f^m({\bf g}^s(x)) = f^m({\bf g}(x))$. We showed that if ${\bf g}\in\mathbb H^\ast$
and $({\bf g},{\bf h})\in E$ then ${\bf h}\in\mathbb H^\ast$. Since $\mathbb H$ is a
strongly connected component of $(\mathbb G,E)$, we conclude that $\mathbb H=\mathbb H^\ast$.

We showed that $f^m({\bf g}(x))$ is the same for all ${\bf g}\in\mathbb H$.
By Proposition~\ref{prop:Range}(c) there exists ${\bf h}\in\mathbb H$
with ${\bf h}(x)=x$, and therefore $f^m({\bf g}(x))=f^m({\bf h}(x))=f^m(x)$
for all ${\bf g}\in\mathbb H$. Since this holds for any $\mathbb H\in{\sf Sinks}({\mathbb G,E})$,
the claim follows.

\subsection{Proof of Theorem~\ref{th:G}(b)}\label{sec:proof:th:Gb}
We mainly follow an argument from~\cite{tz13:stoc} (although without using the language of Markov chains,
relying on the Farkas lemma instead, as in~\cite{kolmogorov15:power}).

Let $(\mathbb G^\ast,E')$ be the subgraph of $(\mathbb G,E)$ induced by $\mathbb G^\ast$.
For an edge $({\bf g},{\bf h)}\in E'$, define positive weight
$w({\bf g},{\bf h})=\hspace{-5pt}\sum\limits_{s\in\supp(\omega):{\bf g}^s={\bf h}}\hspace{-5pt} \omega(s)$. Note that we have
%\begin{eqnarray}
$
\sum\limits_{{\bf h}:({\bf g},{\bf h})\in E'} w({\bf g},{\bf h}) = 1
$
%\end{eqnarray}
for all ${\bf g}\in\mathbb G^\ast$.

We claim that there exists vector $\lambda\in\mathbb R^{\mathbb G^\ast}_{\ge 0}$ that satisfies
\begin{subequations}\label{eq:proof:th:G:b:lambda}
\begin{eqnarray}
\sum_{{\bf g}:({\bf g},{\bf h})\in E'}  w({\bf g},{\bf h})\lambda_{\bf g} -\lambda_{\bf h}&=& 0\qquad\forall {\bf h}\in\mathbb G^\ast \label{eq:proof:th:G:b:lambda:a} \\
\sum_{{\bf g}\in\mathbb G^\ast}\lambda_{\bf g} &=& 1 \label{eq:proof:th:G:b:lambda:b}
\end{eqnarray}
\label{eq:GAISFLAISHGA}
\end{subequations}
Indeed, suppose system~\eqref{eq:GAISFLAISHGA} does not have a solution. By Farkas Lemma (see Lemma~\ref{lemma:Farkas}),
there exists a vector $y\in\mathbb R^{\mathbb G^\ast}$ and a scalar $z\in\mathbb R$ such that
\begin{subequations}
\begin{eqnarray}
z-y_{\bf g}+\sum_{{\bf h}:({\bf g},{\bf h})\in E'}  w({\bf g},{\bf h})y_{\bf h}&\ge& 0\qquad \forall {\bf g}\in\mathbb G^\ast \\
z & < & 0 \label{eq:PAOHGASJFHKAUSG}
\end{eqnarray}
\end{subequations}
Consider ${\bf g}\in\mathbb G^\ast$ with the maximum value of $y_{\bf g}$.
We have
$$
0\le z-y_{\bf g}+\sum_{{\bf h}:({\bf g},{\bf h})\in E'}  w({\bf g},{\bf h})y_{\bf h}\le
z-y_{\bf g}+\sum_{{\bf h}:({\bf g},{\bf h})\in E'}  w({\bf g},{\bf h})y_{\bf g}=z-y_{\bf g}+y_{\bf g}=z
$$
This contradicts~\eqref{eq:PAOHGASJFHKAUSG}, and thus proves that vector $\lambda\ge 0$
satisfying~\eqref{eq:proof:th:G:b:lambda} exists. Next, we will show that this vector satisfies
the property of Theorem~\ref{th:G}(b).

Let us rewrite condition~\eqref{eq:G:omega'} as follows:
\begin{eqnarray}
\sum_{s\in\supp(\omega)} \omega(s)f(x^{{\bf g}^si})&\le& \frac{1}{m-1}\sum_{j\in[m]-\{i\}}f(x^{{\bf g}j})\qquad \forall {\bf g}\in\mathbb G^\ast,i\in[m]
\end{eqnarray}
Multiplying this inequality by $\lambda_{\bf g}$ and summing over ${\bf g}\in\mathbb G^\ast$
(for a fixed $i\in[m]$) gives
\begin{eqnarray}
\sum_{{\bf g}\in\mathbb G^\ast}\sum_{{\bf h}:({\bf g},{\bf h})\in E'}w({\bf g},{\bf h})\lambda_{\bf g}f(x^{{\bf h}i})
&\le& \frac{1}{m-1}\sum_{{\bf g}\in\mathbb G^\ast}\lambda_{\bf g}\sum_{j\in[m]-\{i\}}f(x^{{\bf g}j})\qquad \forall  i\in[m]
\end{eqnarray}
Rearranging terms gives

\begin{eqnarray}
\sum_{{\bf h}\in\mathbb G^\ast}\left[\sum_{{\bf g}:({\bf g},{\bf h})\in E'}w({\bf g},{\bf h})\lambda_{\bf g}\right]f(x^{{\bf h}i})
&\le& \frac{1}{m-1}\sum_{j\in[m]-\{i\}}\sum_{{\bf g}\in\mathbb G^\ast}\lambda_{\bf g}f(x^{{\bf g}j})\qquad \forall i\in[m]\hspace{-30pt}~
\label{eq:FAILSFIAJSFOASFJAIOSFHA}
\end{eqnarray}
By \eqref{eq:proof:th:G:b:lambda:a} the expression in the square brackets equals $\lambda_{\bf h}$, and therefore~\eqref{eq:FAILSFIAJSFOASFJAIOSFHA}
can be rewritten as
\begin{eqnarray}
f^\lambda_i(x)\le \frac{1}{m-1}\sum_{j\in[m]-\{i\}}f^\lambda_j(x)\qquad\forall i\in[m]
\label{eq:GLAIHGASHGASF}
\end{eqnarray}
Consider index $i\in[m]$ with the maximum value of $f^\lambda_i(x)$.
We have $f^\lambda_i(x)\ge f^\lambda_j(x)$ for all $j\in [m]-\{i\}$,
which together with~\eqref{eq:GLAIHGASHGASF}
gives $f^\lambda_i(x)= f^\lambda_j(x)$ for all $j\in [m]-\{i\}$, as claimed.

%\section*{Acknowledgments}
%We would like to acknowledge the assistance of volunteers in putting
%together this example manuscript and supplement.

\bibliographystyle{siamplain}
\bibliography{lit,csp2}
\end{document}